\documentclass[12pt,a4paper]{article}
\usepackage{amsmath,amsthm,amsfonts,amssymb,bbm,sansmath}
\usepackage{graphicx,psfrag,subfigure,color}
\usepackage{cite}
\usepackage{hyperref}

\numberwithin{equation}{section}
\newcommand{\X}{\mathfrak{X}}
\newcommand{\Y}{\mathfrak{Y}}

\newcommand{\Ai}{\mathrm{Ai}}
\newcommand{\const}{\mathrm{const}}
\newcommand{\Pb}{\mathbbm{P}}
\newcommand{\E}{\mathbbm{E}}
\newcommand{\Id}{\mathbbm{1}}
\newcommand{\e}{\varepsilon}
\newcommand{\I}{{\rm i}}
\newcommand{\dx}{\mathrm{d}}
\newcommand{\T}{\mathcal{T}}
\newcommand{\C}{\mathbb{C}}
\newcommand{\R}{\mathbb{R}}
\newcommand{\N}{\mathbb{N}}
\newcommand{\Z}{\mathbb{Z}}
\renewcommand{\Re}{\mathrm{Re}}
\renewcommand{\Im}{\mathrm{Im}}

\DeclareMathOperator{\Tr}{Tr}
\DeclareMathOperator*{\diag}{diag}
\DeclareMathOperator*{\virt}{virt}

\newtheorem{prop}{Proposition}[section]
\newtheorem{thm}[prop]{Theorem}
\newtheorem{lem}[prop]{Lemma}
\newtheorem{defin}[prop]{Definition}
\newtheorem{cor}[prop]{Corollary}

\newtheorem{cla}[prop]{Claim}

\newtheorem{rem}[prop]{Remark}
\newenvironment{remark}{\begin{rem}\normalfont}{\end{rem}}

\title{Why random matrices share universal processes with interacting particle systems?}
\author{Patrik L.\ Ferrari\footnote{Institute for Applied Mathematics, Bonn University, Endenicher Allee 60, 53115 Bonn, Germany. E-mail: {\tt ferrari@uni-bonn.de}}}

\date{Extended lecture notes of the ICTP minicourse in Trieste,\\[0.5em] September 2013}

\begin{document}
\maketitle
\sloppy
\vfill
\begin{abstract}
In these lecture we explain why limiting distribution function, like the Tracy-Widom distribution, or limit processes, like the Airy$_2$ process, arise both in random matrices and interacting particle systems. The link is through a common mathematical structure on an interlacing structure, also known as Gelfand-Tsetlin pattern, that appears for specific models in both fields.
\end{abstract}

\vfill

\newpage
\tableofcontents

\newpage

\section{Introduction}\label{SectIntro}
Universal distribution functions and stochastic processes like the GUE Tracy-Widom distribution and the Airy$_2$ process arise both in random matrix models and some in one-dimensional driven interacting particle systems. Here we explain in which scaling this happens for the Gaussian Unitary Ensemble and its dynamical version given by Dyson's Brownian motion. Then we will do the same for the totally asymmetric simple exclusion process.

\subsection{Random matrices}
\subsubsection{The Gaussian Unitary Ensemble}
Let us first define the random matrix ensemble that we will consider in these lectures.
\begin{defin}
The \textbf{Gaussian Unitary Ensemble} (GUE) of random matrices consists in $N\times N$ Hermitian matrices distributed according to
\begin{equation}\label{eqGUE}
\Pb(H\in \dx H) = \const\exp\left(-\frac{1}{2N}\Tr(H^2)\right) \dx H,
\end{equation}
where the reference measure $\dx H$ is the product Lebesgue measure over the (not constrained by symmetry) entries of the matrix $H$, namely \mbox{$\dx H =\prod_{i=1}^N \dx H_{i,i} \prod_{1\leq i<j\leq N} \dx \Re H_{i,j} \dx \Im H_{i,j}$}, and $\const$ is the normalisation constant.
\end{defin}
The name GUE derives from the fact that the measure (\ref{eqGUE}) is invariant over the unitary transformations and has a Gaussian form.

The GUE measure has the nice (and special) property that, on top of being unitary-invariant, the entries of the matrices are independent. Indeed, another way to define (\ref{eqGUE}) is to set the upper-triangular entries of the matrix $H$ to be independent and Gaussian distributed with $H_{i,j}~{\cal N}(0,N)$ for \mbox{$1\leq i \leq N$}, and $\Re H_{i,j}\sim {\cal N}(0,N/2)$, $\Im H_{i,j}\sim {\cal N}(0,N/2)$, for $1\leq i < j \leq N$.

The prefactor $1/2N$ in front of the trace in (\ref{eqGUE}) can be freely chosen: the only difference is a global scaling of the eigenvalues. With the choice in (\ref{eqGUE}), the largest eigenvalue of a $N\times N$ matrix has fluctuations of order $N^{1/3}$ around the deterministic value $2N$. Also, the eigenvalues' density remains of order $1$ in the bulk of the spectrum. Another standard choice in random matrix literature consists in choosing the prefactor to be $1$, then the eigenvalues are scaled down by a factor $1/\sqrt{2N}$, i.e., the fluctuations are of order $N^{-1/6}$ around $\sqrt{2N}$. This choice is natural since the law of the entries are independent of $N$. Finally, setting the prefactor to be $N$, the largest eigenvalue fluctuates in a scale $N^{-2/3}$ around $\sqrt{2}$. This choice is natural if one wants to study the support of the spectrum, that remains bounded in the $N\to\infty$ limit.

The unitary invariance of (\ref{eqGUE}) allows to compute explicitly the measure of the eigenvalues, because the measure (\ref{eqGUE}) written in terms of the eigenvalues and of the angular variables becomes a product measure. The result is the following.
\begin{lem}\label{lemGUEonePt}
Let $\lambda=(\lambda_1,\ldots,\lambda_N)\in\R^N$ denote the $N$ eigenvalues of a GUE random matrix. Then,
\begin{equation}\label{eqEVdistrGUE}
\Pb(\lambda\in \dx\lambda)=\const\, \Delta_N(\lambda)^2\prod_{i=1}^N e^{-\lambda_i^2/2N}\dx\lambda_i,
\end{equation}
where $\Delta_N(\lambda)=\prod_{1\leq i < j \leq N}(\lambda_j-\lambda_i)\equiv \det(\lambda_i^{j-1})_{i,j=1}^N$ is the Vandermonde determinant and $\const$ is the normalization constant.
\end{lem}

The fluctuations of the largest eigenvalue for GUE were characterized by Tracy and Widom:
\begin{thm}[Tracy and Widom~\cite{TW94}]
Let us denote by $\lambda_{N,\rm max}$ the largest eigenvalue of a $N\times N$ GUE matrix. Then,
\begin{equation}
F_2(s):=\lim_{N\to\infty}\Pb\left(\frac{\lambda_{N,\rm max}-2 N}{N^{1/3}}\leq s\right)
\end{equation}
is a non-degenerate distribution function, called \textbf{GUE Tracy-Widom distribution}. It can be characterized by the Fredholm determinant
\begin{equation}\label{eq1.4}
\begin{aligned}
F_2(s)&=\det(\Id-K_2)_{L^2((s,\infty))}\\
&\equiv \sum_{n\geq 0}\frac{(-1)^n}{n!}\int_s^\infty \dx x_1\cdots \int_s^\infty \dx x_n \det(K_2(x_i,x_j))_{i,j=1}^n
\end{aligned}
\end{equation}
with $K_2(x,y)=\int_{0}^\infty \dx \mu \Ai(x+\mu)\Ai(y+\mu)$, the \textbf{Airy kernel}. $\Ai$ is the Airy function, the solution of $\Ai''(x)= x \Ai(x)$ with $\Ai(x)\sim e^{-2 x^3/3}$ as $x\to\infty$. Equivalently, $F_2$ can be written as
\begin{equation}
F_2(s)=\exp\left(-\int_s^\infty (x-s)^2 q^2(x)\dx x\right),
\end{equation}
where $q$ is the solution of the Painlev\'e II equation $q''(x)= x q(x)+2 q(x)^3$ satisfying the asymptotic condition $q(x)\sim \Ai(x)$ as $x\to\infty$.
\end{thm}

\subsubsection{Dyson's Brownian Motion}
We will see in the course of the lectures where the Fredholm determinant expression comes from, but before let us shortly discuss a dynamics on random matrices introduced by Dyson~\cite{Dys62} and therefore known as \textbf{Dyson's Brownian motion} (DBM). Dyson observed that if the independent entries of a GUE-distributed random matrix evolve as independent stationary Ornstein-Uhlenbeck processes, then the transition probability on matrices is given by
\begin{equation}\label{eq1.6}
\Pb(H(t)\in \dx H | H(0)=H_0)=\const\, \exp\left(-\frac{\Tr(H-q H_0)^2}{2N(1-q^2)}\right)\dx H,
\end{equation}
with $q(t)=e^{-t/2N}$. Further, the evolution of its eigenvalues satisfies the coupled stochastic differential equations
\begin{equation}
\dx \lambda_j(t)=\left(-\frac{1}{2N} \lambda_j(t)+\sum_{i\neq j}\frac{1}{\lambda_j(t)-\lambda_i(t)}\right)\dx t +\dx b_j(t),\quad 1\leq j \leq N,
\end{equation}
where $b_1,\ldots,b_N$ are independent standard Brownian motions. Equivalently, let $L$ be the generator of $N$ independent Ornstein-Uhlenbeck processes,
\begin{equation}
(L f)(\lambda)=\left(\sum_{i=1}^N \frac12 \frac{\partial^2}{\partial \lambda_i^2}-\frac{\lambda_i}{2N}\right) f(\lambda).
\end{equation}
Let $h(\lambda)=\Delta_N(\lambda)$. Then, the generator $L^h$ of the DBM eigenvalues' process is given by
\begin{equation}\label{eqHTransformForDBM}
\begin{aligned}
(L^h f)(\lambda) &= \left(\sum_{i=1}^N \frac12 \frac{\partial^2}{\partial \lambda_i^2}+\left(\sum_{j\neq i}\frac{1}{\lambda_i-\lambda_j}-\frac{1}{2N}\lambda_i\right)\frac{\partial}{\partial \lambda_i}\right) f(\lambda)\\
&=\frac{1}{h(\lambda)} (L (h f))(\lambda).
\end{aligned}
\end{equation}
Using the Harish-Chandra--Itzykson--Zuber formula~\cite{HC57,IZ80} (see (\ref{eqIZHC1})) one then obtains an expression for the joint measure of eigenvalues (an extension of Lemma~\ref{lemGUEonePt}). It turns out that the joint distributions of the largest eigenvalue at different times is given as a Fredholm determinant too.

In the large-$N$ limit the process of the largest eigenvalue converges (properly rescaled) to the Airy$_2$ process, ${\cal A}_2$:
\begin{thm}
Let us denote by $\lambda_{N,\rm max}(t)$ the largest eigenvalue of the stationary GUE Dyson's Brownian motion. Then,
\begin{equation}
\lim_{N\to\infty} \Pb\left(\bigcap_{k=1}^m\left\{\frac{\lambda_{N,\rm max}(2 t_k N^{2/3})-2N}{N^{1/3}}\leq s_k\right\}\right) = \Pb\left(\bigcap_{k=1}^m \left\{{\cal A}_2(t_k)\leq s_k\right\}\right),
\end{equation}
where ${\cal A}_2$ is the \textbf{Airy$_2$ process}. It is defined by its joint-distributions: for any $m\in\N$, $t_1<t_2<\ldots<t_m\in\R$ and $s_1,\ldots, s_m\in\R$, it holds
\begin{equation}\label{eq1.11}
\begin{aligned}
&\Pb\left(\bigcap_{k=1}^m \left\{{\cal A}_2(t_k)\leq s_k\right\}\right)=\det(\Id-P_s K_2 P_s)_{L^2(\R\times \{t_1,\ldots,t_m\})}\\
&\equiv \sum_{n\geq 0}\frac{(-1)^n}{n!}\sum_{\ell_1=1}^m\cdots \sum_{\ell_n=1}^m\int_{s_{\ell_1}}^\infty \dx x_1\cdots\int_{s_{\ell_n}}^\infty \dx x_n \det(K_2(x_i,t_{\ell_i};x_j,t_{\ell_j}))_{i,j=1}^n,
\end{aligned}
\end{equation}
where $P_s(x,t_k)=\Id_{[x\leq s_k]}$ and the \textbf{extended Airy kernel} $K_2$ is given by
\begin{equation}\label{eq1.12}
K_2(x,t;x',t')=\left\{
                 \begin{array}{ll}
                   \int_{\R_+}\dx\lambda e^{-\lambda(t'-t)}\Ai(x+\lambda)\Ai(x'+\lambda), & \textrm{for }t\leq t', \\
                   -\int_{\R_-}\dx\lambda e^{-\lambda(t'-t)}\Ai(x+\lambda)\Ai(x'+\lambda), & \textrm{for }t>t'.
                 \end{array}
               \right.
\end{equation}
\end{thm}
This result is known to be true since about 10 years but a proof was not written down explicitly for a while (being just an exercise in comparison for instance to~\cite{Jo01}). The convergence of the kernel can be found for example in Appendix~A of~\cite{FerPhD}, of the process in~\cite{Wei11}.

\subsection{The totaly asymmetric simple exclusion process}
The \textbf{totally asymmetric simple exclusion process} (TASEP) is one of the simplest
interacting stochastic particle systems. It consists of particles on the
lattice of integers, $\Z$, with at most one particle at each site (exclusion principle).
The dynamics in continuous time is as follows. Particles jump on the
neighboring right site with rate $1$ provided that the site is empty. This means
that jumps are independent of each other and take place after an exponential
waiting time with mean $1$, which is counted from the time instant when the
right neighbor site is empty.

More precisely, we denote by $\eta$ a particle configuration, $\eta\in
\Omega=\{0,1\}^\Z$. Let $f$: $\Omega\to \R$ be a function depending only on a
finite number of $\eta_j$'s. Then the backward generator of the TASEP is given by
\begin{equation}\label{1.1}
Lf(\eta)=\sum_{j\in\Z}\eta_j(1-\eta_{j+1})\big(f(\eta^{j,j+1})-f(\eta)\big).
\end{equation}
Here $\eta^{j,j+1}$ denotes the configuration $\eta$ with the
occupations at sites $j$ and \mbox{$j+1$} interchanged. The semigroup $e^{Lt}$ is
well-defined as acting on bounded and continuous functions on $\Omega$. $e^{Lt}$
is the transition probability of the TASEP~\cite{Li99}.

If we denote by $\rho(\xi,\tau)$ the macroscopic density at position $\xi\in\R$ and time $\tau\in\R$, given by $\lim_{\e\to 0} \Pb(\eta_{[\xi \e^{-1}]}(\tau \e^{-1})=1)$, then it satisfies the (deterministic) Burgers equation
\begin{equation}
\partial_\tau \rho+\partial_\xi (\rho(1-\rho))=0.
\end{equation}

TASEP dynamics preserves the order of particles. We denote by $x_k(t)$ the position of particle with label $k$ at time $t$ and choose the right-to-left ordering, i.e., $x_{k+1}(t)<x_k(t)$ for any $k$. Consider now the initial condition where the left of $0$ is initially fully occupied and the rest is empty, i.e., $x_k(0)=-k$ for $k\geq 1$. This is called \textbf{step initial condition}.
\begin{thm}[Johansson~\cite{Jo00b}]
Consider TASEP with step initial condition. Then,
\begin{equation}
\lim_{t\to\infty}\Pb(x_{[t/4]}(t)\geq -s (t/2)^{1/3})=F_2(s).
\end{equation}
\end{thm}
Further, the joint law of particles' positions is governed, in the large time $t$ limit, by the Airy$_2$ process:
\begin{thm}\label{thmTASEP}
It holds
\begin{equation}
\lim_{t\to\infty} \frac{x_{[t/4+u(t/2)^{2/3}]}(t)+2u(t/2)^{2/3}-u^2(t/2)^{1/3}}{-(t/2)^{1/3}}={\cal A}_2(u)
\end{equation}
in the sense of finite-dimensional distributions.
\end{thm}
The first proof of Theorem~\ref{thmTASEP} goes back to~\cite{Jo00b} for the one-point and it was improved to a functional limit theorem in a discrete time setting in~\cite{Jo01}. To be precise, the result was for a last passage percolation. However, along the characteristics of the Burgers equation, the decorrelation happens in a macroscopic scale (one has to move away of order $t$ to see a fluctuation of order $t^{1/3}$), in particular it is much slower than in the spatial direction (where it is enough to move away of order $t^{2/3}$). This is phenomenon is also known as \textbf{slow-decorrelation}~\cite{Fer08,CFP10b}. Using this fact one can deduce Theorem~\ref{thmTASEP} for joint distributions from the last passage percolation result as shown in~\cite{CFP10b}.

\subsection{Outlook}
In these lectures we will address the question of the reason why the Tracy-Widom distribution and the Airy$_2$ process arises both in the GUE model of random matrices and in TASEP. In Section~\ref{SectIntPart} we define an interacting particle system on interlacing configuration evolving as Markov chain. TASEP will be a projection of this model. The measures that we will encounter are \mbox{$L$-ensembles} and have determinantal correlation functions, see Section~\ref{SectLmeasures}. In Section~\ref{SectRandMatr} we will come back to random matrices and discover an interlacing structure as in the particle system.

\section{Interacting particle systems}\label{SectIntPart}
In this section we consider a Markov chain on interlacing configurations (also known as \textbf{Gelfand-Tsetlin patterns}). The state space is given by
\begin{equation}
\small GT_N=\left\{X^N=(x^1,\ldots,x^N), x^n=(x_1^n,\ldots,x_n^n)\in\Z^n\, |\, x^n\prec x^{n+1}, 1\leq n\leq N-1\right\},
\end{equation}
where
\begin{equation}
x^n\prec x^{n+1}\iff x_1^{n+1} < x_1^n \leq x_2^{n+1} < x_2^{n+1}\leq \ldots < x_n^n \leq x_{n+1}^{n+1}.
\end{equation}
If $x^n\prec x^{n+1}$ we say that \textbf{$x^n$ interlace with $x^{n+1}$}. We can (and will) think of the configurations as unions of levels: so the vector $x^n$ of $X^N$ is the state at \emph{level $n$}. See Figure~\ref{figGTN} for an illustration.
\begin{figure}[t]
\begin{center}
  \psfrag{m}[c]{$<$} \psfrag{e}[c]{$\leq$} \psfrag{x11}{$x_1^1$}
  \psfrag{x12}{$x_1^2$} \psfrag{x13}{$x_1^3$} \psfrag{x14}{$x_1^4$}
  \psfrag{x22}{$x_2^2$} \psfrag{x23}{$x_2^3$} \psfrag{x24}{$x_2^4$}
  \psfrag{x33}{$x_3^3$} \psfrag{x34}{$x_3^4$} \psfrag{x44}{$x_4^4$}
  \includegraphics[height=3.5cm]{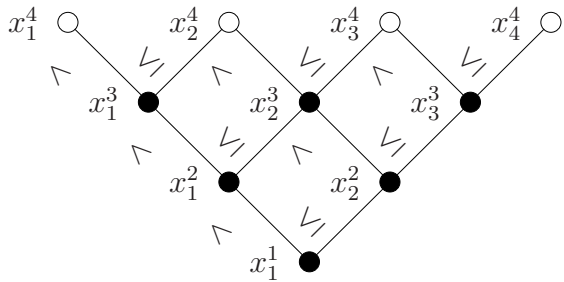}
\caption{Graphical representation of $GT_4$. The white dots represents the vector at level $n=4$.}
\label{figGTN}
\end{center}
\end{figure}

The Markov chain is built up through two basic Markov chains: (a) the first is the time evolution at a fixed level and (b) the second is a Markov chain on $GT_N$, linking level $n$ with level $n-1$, $2\leq n\leq N$. We first discuss the two chains separately and then define the full Markov chain with their properties.

\subsection{The Charlier process}\label{sectCharlier}
Here we consider a discrete analogue of Dyson's Brownian motion, where the Ornstein-Uhlenbeck processes are replaced by Poisson processes of intensity~$1$ (illustrated in Figure~\ref{figCharlier}).
\begin{figure}[t]
\begin{center}
  \psfrag{x}{$x$}
  \psfrag{t}{$t$}
  \includegraphics[height=3.5cm]{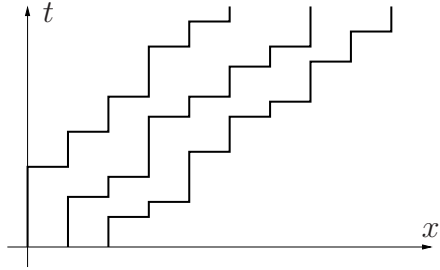}
\caption{Possible space-time trajectories for the Charlier process with $n=3$ particles.}
\label{figCharlier}
\end{center}
\end{figure}

Let us fix $n\in \N$. Consider the continuous-time random walk \mbox{$x^n=(x_1^n,\ldots,x_n^n)\in\Z^n$} where each of the $x_k^n$, $1\leq k \leq n$, are independent one-sided random walks with jump rate $1$. The generator of the random walk is
\begin{equation}
(L_n f)(x)=\sum_{i=1}^n \nabla_i f(x),
\end{equation}
where $\nabla_i f(x)= f(x+e_i)-f(x)$, with $e_i$ the vector with entries $e_i(j)=\delta_{i,j}$. We want to condition the random walk $x^n$ on never having a collision between any two of its components, i.e., to stay in the Weyl chamber
\begin{equation}
W_n=\{x=(x_1,\ldots,x_n)\in\Z^n\, | \, x_1<x_2<\ldots<x_n\}.
\end{equation}
This can be achieved by the Doob $h$-transform with $h$ given by the Vandermonde determinant~\cite{KOR02}. Indeed, one can verify that
\begin{equation}\label{eq2.5}
h_n(x)=\prod_{1\leq i < j \leq n} (x_j-x_i)\equiv \Delta_n(x)
\end{equation}
is harmonic, i.e., $L_n h_n = 0$. Then, the conditioned process given that the walker never leaves the Weyl chamber $W_n$ is the Doob $h$-transform of the free walk. This process is called \textbf{Charlier process}. For $x,y\in W_n$ and $t>0$, the transition probability $P_{n,t}$ from $x$ to $y$ of the conditioned random walk is given by
\begin{equation}
P_{n,t}(x,y)=\frac{h_n(y)}{h_n(x)}\Pb(x^n(t)=y, T>t\,|\, x^n(0)=x),
\end{equation}
where $T=\inf\{t>0\, |\, x^n(t)\not\in W_n\}$. Using Karlin-McGregor's formula~\cite{KM59}, we also have
\begin{equation}
\Pb(x^n(t)=y, T>t\,|\, x^n(0)=x) = \det\left(p_t(y_i-x_j)\right)_{i,j=1}^n,
\end{equation}
with
\begin{equation}
p_t(x)= \frac{e^{-t}t^x}{x!}\Id_{x\geq 0}
\end{equation}
the one-particle transition probability. We have obtained the following:
\begin{prop}\label{propCharlier}
The Charlier process has the transition probability given by
\begin{equation}\label{eqTransProbaCharlier}
P_{n,t}(x,y)=\frac{\Delta_n(y)}{\Delta_n(x)}\det\left(p_t(y_j-x_i)\right)_{i,j=1}^n
\end{equation}
with $x,y\in W_n$, $t>0$, and $\Delta_n$ the Vandermonde determinant.
\end{prop}

This conditioned process has the generator $L^h_n$ given by
\begin{equation}\label{eqGenCharlier}
(L^h_n f)(x)=\frac{1}{h_n(x)} (L_n (h_n f))(x).
\end{equation}
One clearly sees the similarity between (\ref{eqGenCharlier}) and (\ref{eqHTransformForDBM}).

The reason for the name Charlier process is the following. Consider the initial condition $x^n(t=0)=x^*\equiv(0,1,\ldots,n-1)$. Then, one can show (see Proposition~3.3 of~\cite{KOR02}) that
\begin{equation}\label{eqDensityCharlier}
\Pb_{n,t}(x^*,x) = \const\, \Delta_n(x)^2\prod_{i=1}^n\frac{e^{-t} t^{x_i}}{x_i!},
\end{equation}
for some normalization constant $\const$. A measure of the form (\ref{eqDensityCharlier}) has, as explained in Section~\ref{SectLmeasures}, determinantal correlation functions expressed in terms of Charlier orthogonal polynomials.

\subsubsection*{A discrete time version}
The discrete time analogue is obtained by setting the one-step transition probability of one-particle by\footnote{For a set $S$, by $\frac{1}{2\pi\I}\oint_{\Gamma_S} \dx w f(w)$ we mean the contour integral where the contour can be taken to be any anticlockwise oriented simple path containing all the points in the set $S$ but no other poles of the function $f$.}
\begin{equation}\label{eq2.11}
P(x,y)= \frac{1}{2\pi\I}\oint_{\Gamma_0} \dx w \frac{1-p+p w^{-1}}{w^{x-y+1}}
=\left\{
         \begin{array}{ll}
           p, & \textrm{if }y=x+1, \\
           1-p, & \textrm{if }y=x, \\
           0, & \textrm{otherwise}.
         \end{array}
       \right.
\end{equation}
This discrete time analogue of the Charlier process has the transition probability given by
\begin{equation}
p_t(x)=(P^t)(0,x)=\left(\begin{array}{c}t\\x\end{array}\right)p^x (1-p)^{t-x}\Id_{[0\leq x\leq t]}
\end{equation}
for $t\in\N$. The Charlier process is then recovered by replacing $t$ by $t/p$ and taking the $p\to 0$ limit.

Then Proposition~\ref{propCharlier} is still valid and becomes.
\begin{prop}\label{propDiscreteCharlier}
The discrete time analogue of the Charlier process has the one-step transition probability given by
\begin{equation}\label{eq2.14}
P_{n}(x,y)=\frac{\Delta_n(y)}{\Delta_n(x)}\det\left(P(x_i,y_j)\right)_{i,j=1}^n
\end{equation}
with $P$ as in (\ref{eq2.11}), $x,y\in W_n$, and $\Delta_n$ the Vandermonde determinant.
\end{prop}

\subsection{The interlacing Markov link}\label{sectLevelLinks}
Now we consider a Markov link between levels of $GT_N$ that generates the uniform measure on $GT_N$ given the value $x^N$ of the level $N$, i.e., with \mbox{$x^N=(x_1^N<x_2^N<\ldots<x_N^N)$} fixed. It can be shown, see Corollary~\ref{CorA4}, that
\begin{equation}\label{eq2.12}
\textrm{\# of $GT_N$ patterns with given $x^N$}=\prod_{1\leq i < j \leq N}\frac{x^N_j-x_i^N}{j-i} = \frac{\Delta_N(x^N)}{\prod_{n=1}^{N-1}n!}.
\end{equation}
Thus, the uniform measure on $GT_N$ given $x^N$ can be obtained by setting
\begin{equation}
\Pb(x^{N-1}\, | \, x^N)=\frac{\textrm{\# of $GT_{N-1}$ patterns with given $x^{N-1}$}}{\textrm{\# of $GT_N$ patterns with given $x^N$}}\Id_{[x^{N-1}\prec x^N]}.
\end{equation}
Using (\ref{eq2.12}) we obtain
\begin{equation}
\Pb(x^{N-1}\, | \, x^N) = (N-1)! \frac{\Delta_{N-1}(x^{N-1})}{\Delta_N(x^N)}\Id_{[x^{N-1}\prec x^N]}.
\end{equation}
Consequently, let us define the Markov link between level $n$ and $n-1$ by
\begin{equation}\label{eqMarkovLink}
\Lambda^n_{n-1}(x^n,x^{n-1}):=(n-1)!\frac{\Delta_{n-1}(x^{n-1})}{\Delta_n(x^n)}\Id_{[x^{n-1}\prec x^n]}.
\end{equation}
for $n=2,\ldots,N$. Then, the measure on $GT_N$ given $x^N$ is given by
\begin{equation}\label{eq2.13}
\prod_{n=2}^N (n-1)!\frac{\Delta_{n-1}(x^{n-1})}{\Delta_n(x^n)}\Id_{[x^{n-1}\prec x^n]} = \frac{\prod_{n=1}^{N-1}n!}{\Delta_N(x^N)} \Id_{[x^1\prec x^2\prec\ldots\prec x^N]},
\end{equation}
i.e., it is the uniform measure on $GT_N$ given $x^N$ by (\ref{eq2.12}).

There is an important representation of the interlacing through determinants. This will be relevant when studying more in details the correlation functions, see Section~\ref{SectLmeasures}.
\begin{lem}\label{LemInterlacing}
Let $x^N\in W_N$ and $x^{N-1}\in W_{N-1}$ be ordered configurations in the Weyl chambers. Then, setting $x_N^{N-1}\equiv {\rm virt}$ a ``virtual variable'', we have
\begin{equation}\label{eqInterlacingDeterminants}
\Id_{[x^{N-1}\prec x^N]} = \pm \det(\phi(x_i^{N-1},x_j^N))_{i,j=1}^N,
\end{equation}
with $\phi(x,y)=\Id_{[y\geq x]}=\frac{1}{2\pi\I}\oint_{\Gamma_0} \dx w \frac{(1-w)^{-1}}{w^{y-x+1}}$ and $\phi({\rm virt},y)=1$ (the $\pm$ sign depends on the size of the matrix).
\end{lem}
\begin{proof}
The proof is quite easy. For $x^{N-1}\prec x^N$ one sees that the matrix on the r.h.s.~of (\ref{eqInterlacingDeterminants}) is triangular with $1$ on the diagonal. Further, by violating the interlacing conditions, one gets two rows or columns which are equal, so that the determinant is equal to zero.
\end{proof}

\subsection{Intertwining of the Markov chains}\label{sectIntertwining}
In Section~\ref{sectCharlier} we have described a continuous time Markov chain living on the Weyl chamber $W_n$, for any $n\in\N$, while in Section~\ref{sectLevelLinks} we have constructed a Markov link between states in $W_n$ and $W_{n-1}$ by the transition kernel (\ref{eqMarkovLink}). In this section we want to define a Markov chain on $W_n\times W_{n-1}$ such that its projections on $W_n$ and $W_{n-1}$ are Charlier processes and projection at fixed time is the process given by $\Lambda^n_{n-1}$. This can then be easily extended to a Markov chain on the whole $GT_N$. It is simpler to understand the construction in discrete time. Therefore we will first do it for the discrete time analogue of the Charlier process. We can then take the continuous limit afterwards on the main statement. The construction discussed here is a particular case of the one in~\cite{BF08,DF90}.

The key property that will allow us to define such a dynamic is the \textbf{intertwining relation}
\begin{equation}\label{eqIntertwining}
\Delta^{n}_{{n-1}}:=P_{n} \Lambda^n_{n-1} = \Lambda^n_{n-1} P_{n-1}\quad 2\leq n \leq N.
\end{equation}
In our specific case, to see that (\ref{eqIntertwining}) one uses the Fourier representation for $P_n$ (see (\ref{eq2.11})) and of $\Lambda^n_{n-1}$ (see Lemma~\ref{LemInterlacing}). The intertwining relation can be obtained quite generically when the transition matrices are translation invariant, see Proposition 2.10 of~\cite{BF08} for a detailed statement. See Appendix~\ref{AppToeplitz} for more details.

\subsubsection{Construction in discrete time}
Let us now explain the generic construction. Let $P_n$ be the transition probability of a Markov chain in $S_n$ and let $\Lambda^n_{n-1}$ be a Markov link between $S_n$ and $S_{n-1}$ satisfying the intertwining condition (\ref{eqIntertwining}), illustrated in Figure~\ref{figIntertwining}.
\begin{figure}[t!]
\begin{center}
  \psfrag{L}[l]{$\Lambda^{n}_{n-1}$}
  \psfrag{Pn}[c]{$P_n$}
  \psfrag{Pnm1}[c]{$P_{n-1}$}
  \psfrag{Sn}[c]{$S_n$}
  \psfrag{Snm1}[c]{$S_{n-1}$}
  \psfrag{C}[c]{$\circlearrowleft$}
  \includegraphics[height=3cm]{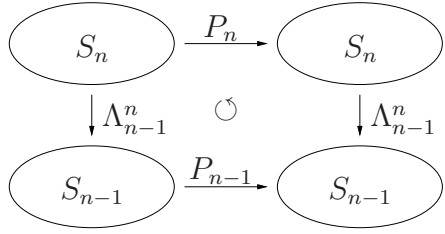}
\caption{Intertwinings.}
\label{figIntertwining}
\end{center}
\end{figure}

Denote by
\begin{equation}
S_\Lambda^n=\{X^n=(x^1,\ldots,x^n)\in S_1\times\cdots\times S_n\, | \, \Lambda^k_{k-1}(x^k,x^{k-1})\neq 0, 2\leq k \leq n\},
\end{equation}
the set of allowed configurations. In our special case, $S_n=W_n$ and $S_\Lambda^n$ is nothing else than $GT_n$.

Define the transition probabilities of a Markov chain on $S_\Lambda^n$ by (we use the notation $X^n=(x^1,\ldots,x^n)$ and $Y^n=(y^1,\ldots,y^n)$)
\begin{multline}\label{eqMConGTn}
P_{\Lambda}^n(X^n,Y^n)\\
=\left\{
\begin{array}{ll}
P_{1}(x^1,y^1)\prod_{k=2}^n\frac{P_{k}(x^k,y^k)\Lambda^k_{k-1}(y^k,y^{k-1})}{\Delta^k_{k-1}(x^k,y^{k-1})}, & \prod_{k=2}^n\Delta^k_{k-1}(x^k,y^{k-1})>0, \\
0, & \textrm{otherwise}.
\end{array}
\right.
\end{multline}
One can think of $P_{\Lambda}^n$ as follows. Starting from $X^n=(x^1,\ldots,x^n)$, we first choose $y^1$ according to the transition matrix $P_{1}(x^1,y^1)$, then choose $y^2$ using $\frac{P_{2}(x^2,y^2)\Lambda^2_{1}(y^2,y^{1})}{\Delta^2_{1}(x^2,y^{1})}$, which is the conditional distribution of the middle point in the successive application of $P_{2}$ and $\Lambda^2_1$, provided that we start at $x^2$ and finish at $y^1$. After that we choose $y^3$ using the conditional distribution of the middle point in the successive application of $P_{3}$ and $\Lambda^3_2$ provided that we start at $x^3$ and finish at $y^2$, and so on. This is called \emph{sequential update}.

With our specific choice of $P_n$'s and $\Lambda^n_{n-1}$'s the dynamics is the following:
\begin{itemize}
  \item[(a)] $x_1^1$ just performs a one-sided random walk (with jump probability $p$).
  \item[(b1)] $x_1^2$ performs a one-sided random walk but the jumps leading to $x_1^2=x_1^1$ are suppressed (we say that $x_1^2$ is blocked by $x_1^1$).
  \item[(b2)] $x_2^2$ performs one-sided random walk but the jumps leading to $x_2^2=x_1^1$ are forced to happen (we say that $x_2^2$ is pushed by $x_1^1$).
  \item[(c)] Similarly, $x_k^n$ is blocked by $x_k^{n-1}$ and is pushed by $x_{k-1}^{n-1}$ (whenever they exists).
\end{itemize}

\subsubsection{A class of conserved measures}\label{sectConsDetMeas}
Given the Markov chain on $S_\Lambda^n$ described above, it is of interest to know which class of measures are conserved by the time evolution. Here is such a class.

\begin{prop}\label{PropConservationMeasureType}
Let $\mu_n(x^n)$ a probability measure on $S_n$. Consider the evolution of the measure
\begin{equation}\label{eq2.20}
M_n(X^n)=\mu_n(x^n)\Lambda^n_{n-1}(x^n,x^{n-1})\Lambda^{n-1}_{n-2}(x^{n-1},x^{n-2})\cdots\Lambda^2_1(x^2,x^1)
\end{equation}
on $S_\Lambda^n$ under the Markov chain $P^n_{\Lambda}$. Then the measure at time $t$ is given by
\begin{multline}
(M_n \underbrace{P^n_{\Lambda}\cdots P^n_{\Lambda}}_{t\textrm{ times}})(Y^n)\\
=(\mu_n \underbrace{P_n\cdots P_n}_{t\textrm{ times}})(y^n)\Lambda^n_{n-1}(y^n,y^{n-1})\Lambda^{n-1}_{n-2}(y^{n-1},y^{n-2})\cdots\Lambda^2_1(y^2,y^1).
\end{multline}
\end{prop}
\begin{proof} It is enough to prove it for $t=1$.
The measure (\ref{eq2.20}) evolved by $P^n_{\Lambda}$ is given by
\begin{equation}\label{eq2.24}
\sum_{x^1,\ldots,x^n} \mu_n(x^n)\prod_{k=2}^n \Lambda^k_{k-1}(x^k,x^{k-1}) P_{1}(x^1,y^1)\prod_{k=2}^n\frac{P_{k}(x^k,y^k)\Lambda^k_{k-1}(y^k,y^{k-1})}{\Delta^k_{k-1}(x^k,y^{k-1})}.
\end{equation}
By the intertwining property (\ref{eqIntertwining}), it holds $\Lambda^2_1 P_{1}=\Delta^2_{1}$ so that
\begin{equation}
\sum_{x^1} \Lambda^2_1(x^2,x^1) P_{1}(x^1,y^1) = \Delta^2_{1}(x^2,y_1).
\end{equation}
This term cancels the denominator for $k=2$ of the last term in (\ref{eq2.24}). Similarly, applying sequentially the sums over $x^2, x^3,\ldots,x^{n-1}$ we obtain
\begin{equation}
\begin{aligned}
(\ref{eq2.24})&=\sum_{x^n}\mu_n(x^n)P_{n}(x^n,y^n)\prod_{k=2}^n\Lambda^k_{k-1}(y^k,y^{k-1})\\
&=(\mu_n P_{n})(y^n)\prod_{k=2}^n\Lambda^k_{k-1}(y^k,y^{k-1}),
\end{aligned}
\end{equation}
that is the claimed result.
\end{proof}

In particular, if we consider the measure $\mu_n$ given by (\ref{eqDensityCharlier}) and $\Lambda^n_{n-1}$ as in (\ref{eqMarkovLink}) then the measure (\ref{eq2.20}) turns out to have determinantal correlations (see Section~\ref{SectLmeasures}). This nice property is conserved by the time evolution.

The next question is to determine the joint measure at different times and different levels, see Figure~\ref{FigDetStructure} for an illustration.
\begin{figure}[t!]
\begin{center}
\psfrag{0}[r]{$0$}
\psfrag{n2}[r]{$n_2$}
\psfrag{n1}[r]{$n_1$}
\psfrag{t0}[c]{$0$}
\psfrag{t1}[c]{$t_1$}
\psfrag{t2}[c]{$t_2$}
\psfrag{112}[r]{$x_1^1(t_2)$}
\psfrag{1n22}[r]{$x_1^{n_2}(t_2)$}
\psfrag{n2n22}{$x_{n_2}^{n_2}(t_2)$}
\psfrag{1n21}[r]{$x_1^{n_2}(t_1)$}
\psfrag{n2n21}{$x_{n_2}^{n_2}(t_1)$}
\psfrag{1n11}[r]{$x_1^{n_1}(t_1)$}
\psfrag{n1n11}{$x_{n_1}^{n_1}(t_1)$}
\psfrag{1n10}[r]{$x_{1}^{n_1}(0)$}
\psfrag{n1n10}{$x_{n_1}^{n_1}(0)$}
\includegraphics[height=6cm]{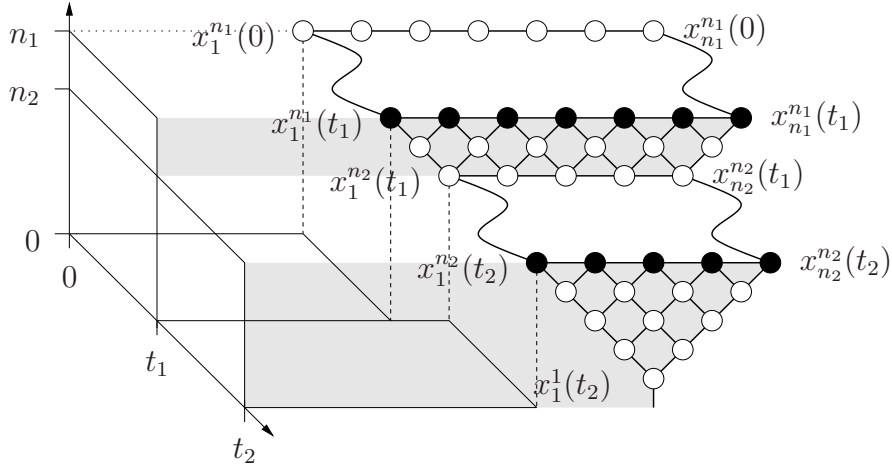}
\caption{A graphical representation of variables entering in Proposition~\ref{PropJointMeasures}, illustrated for $m=2$. The wavy lines represents the time evolution between $0$ and $t_1$ and from $t_1$ to $t_2$. For simplicity we have taken $n_1=n$. The black dots represents the variables we project on.}
\label{FigDetStructure}
\end{center}
\end{figure}

\begin{prop}[Proposition 2.5 of~\cite{BF08}]\label{PropJointMeasures}
Let $\mu_n(x^n)$ a probability measure on $S_n$. Consider the evolution of the measure
\begin{equation}\label{eq2.28a}
\mu_n(x^n)\Lambda^n_{n-1}(x^n,x^{n-1})\Lambda^{n-1}_{n-2}(x^{n-1},x^{n-2})\cdots\Lambda^2_1(x^2,x^1)
\end{equation}
on $S_\Lambda^n$ under the Markov chain $P^n_{\Lambda}$. Denote by $(x^1(t),\ldots,x^n(t))$ the result at time $t$.
Consider $m$ ``space-like''  points $(n_1,t_1)\prec (n_2,t_2)\prec \cdots \prec (n_m,t_m)$ where
\begin{equation}\label{eqDefPartialOrder}
(n_i,t_i)\prec (n_j,t_j) \iff n_i\geq n_j, t_i\leq t_j, \textrm{and }(n_i,t_i)\neq (n_j,t_j).
\end{equation}
With the notation
\begin{equation}\label{eqDefDelta}
\Delta^n_{m,t}:=(P_{n})^t\Lambda^n_{n-1}\cdots\Lambda^{m+1}_{m}.
\end{equation}
for $n>m\geq 1$, the joint distribution of
\begin{equation}
(x^{n_1}(t_1),\ldots,x^{n_m}(t_m))
\end{equation}
coincides with the stochastic evolution of $\mu_n$ under the transition matrices
\begin{equation}
(\Delta^n_{n_1,t_1},\Delta^{n_1}_{n_2,t_2-t_1},\ldots,\Delta^{n_{m-1}}_{n_m,t_m-t_{m-1}}).
\end{equation}
\end{prop}
\begin{proof}
We write the detailed proof for $m=2$. Its extension to generic $m$ is straightforward but a bit lengthly in the notations.
By Proposition~\ref{PropConservationMeasureType} and the definition of the transition probability (\ref{eqMConGTn}), the joint measure at times $t_1$ and $t_2$ is given by
\begin{equation}\label{eq2.27}
\begin{aligned}
&(\mu_n (P_{n})^{t_1})(x^n(t_1))\prod_{k=2}^n\Lambda^k_{k-1}(x^k(t_1),x^{k-1}(t_1))\\
\times & (P_{1})^{t_2-t_1}(x^1(t_1),x^1(t_2))
\prod_{k=2}^n\frac{(P_{k})^{t_2-t_1}(x^k(t_1),x^k(t_2))\Lambda^k_{k-1}(x^k(t_2),x^{k-1}(t_2))}{\Delta^k_{k-1,t_2-t_1}(x^k(t_1),x^{k-1}(t_2))}.
\end{aligned}
\end{equation}
Summing (\ref{eq2.27}) over $x^{n_2+1}(t_2),\ldots,x^n(t_2)$ it results into
\begin{equation}\label{eq2.28}
\begin{aligned}
&(\mu_n (P_{n})^{t_1})(x^n(t_1))\prod_{k=2}^n\Lambda^k_{k-1}(x^k(t_1),x^{k-1}(t_1))\\
\times & (P_{1})^{t_2-t_1}(x^1(t_1),x^1(t_2))
\prod_{k=2}^{n_2}\frac{(P_{k})^{t_2-t_1}(x^k(t_1),x^k(t_2))\Lambda^k_{k-1}(x^k(t_2),x^{k-1}(t_2))}{\Delta^k_{k-1,t_2-t_1}(x^k(t_1),x^{k-1}(t_2))}.
\end{aligned}
\end{equation}
Then summing (\ref{eq2.28}) over $x^{n_1+1}(t_1),\ldots,x^n(t_1)$ and using the definition (\ref{eqDefDelta}) we have
\begin{equation}\label{eq2.29}
\begin{aligned}
&(\mu_n \Delta^n_{n_1,t_1})(x^{n_1}(t_1)) \prod_{k=2}^{n_1}\Lambda^k_{k-1}(x^k(t_1),x^{k-1}(t_1)) \\
\times & (P_{1})^{t_2-t_1}(x^1(t_1),x^1(t_2))
\prod_{k=2}^{n_2}\frac{(P_{k})^{t_2-t_1}(x^k(t_1),x^k(t_2))\Lambda^k_{k-1}(x^k(t_2),x^{k-1}(t_2))}{\Delta^k_{k-1,t_2-t_1}(x^k(t_1),x^{k-1}(t_2))}.
\end{aligned}
\end{equation}
By summing (\ref{eq2.29}) over $x^1(t_1),\ldots,x^{n_2-1}(t_1)$ and then $x^1(t_2),\ldots,x^{n_2-1}(t_2)$ we obtain
\begin{equation}\label{eq2.30}
(\mu_n \Delta^n_{n_1,t_1})(x^{n_1}(t_1))\prod_{k=n_2+1}^{n_1}\Lambda^k_{k-1}(x^k(t_1),x^{k-1}(t_1)) (P_{n_2})^{t_2-t_1}(x^{n_2}(t_1),x^{n_2}(t_2)).
\end{equation}
Finally, summing up (\ref{eq2.30}) over $x^{n_2}(t_1),\ldots,x^{n_1-1}(t_1)$ and using (\ref{eqDefDelta}) together with (\ref{eqIntertwining}) we obtain
\begin{equation}
(\mu_n \Delta^n_{n_1,t_1})(x^{n_1}(t_1)) \Delta^{n_1}_{n_2,t_2-t_1}(x^{n_1}(t_1),x^{n_2}(t_2)),
\end{equation}
which is the claimed result for $m=2$.
\end{proof}

\newpage
\subsubsection{Continuous time analogue}
Consider $S_n=W_n$, $P_n$ as (\ref{eq2.14}) and $\Lambda^n_{n-1}$ as in (\ref{eqMarkovLink}). Then, by taking the continuous time limit, we get that Propositions~\ref{PropConservationMeasureType} and~\ref{PropJointMeasures} still holds for the Charlier case with $(P_n)^t$ replaced by $P_{n,t}$ given in (\ref{eqTransProbaCharlier}).

Below we consider the process arising from the Charlier process at level $N$ starting with $x^N(0)=(-N,-N+1,\ldots,-1)$. Interlacing implies that the initial condition of $GT_N$ given this $x^N(0)$ is the deterministic configuration $x_k^n(0)=-n+k$, $1\leq k \leq n \leq N$. Further, since by construction the dynamics of level $n$ does not depends of the evolution of the level above it, then the evolution of level $n$ is a Charlier process at level $n$ itself.

\begin{prop}\label{PropMeasureContTime}
Consider the process arising from the packed initial condition, $x_k^n(0)=-n-1+k$, $1\leq k \leq n \leq n_1$. Let us consider the joint distributions at $m$ ``space-like'' points $(n_1,t_1)\prec (n_2,t_2)\prec \cdots \prec (n_m,t_m)$. For any level $n$ there is are $c(n)=\# \{i| n_i=n\}\in \{0,\ldots,m\}$ consecutive times in $\{t_1,\ldots,t_m\}$, that we denote by $t_0^n<\ldots<t_{c(n)}^n$. Then, the joint distribution of
\begin{equation}
(x^{n_1}(t_1),\ldots,x^{n_m}(t_m))
\end{equation}
is a marginal of the measure
\begin{multline}\label{eq2.40}
\const\, \prod_{n=1}^{n_1}\bigg[\det[\phi(x_i^{n-1}(t_0^{n-1}),x_j^n(t_{c(n)}^n))]_{i,j=1}^n\\
\times \prod_{a=1}^{c(n)}\det[p_{t_a^n-t_{a-1}^n}(x_i^n(t_{a-1}^n),x_j^n(t_a^n))]_{i,j=1}^n\bigg]\\
\times \Delta_{n_1}(x^{n_1}(t_0)^{n_1})\prod_{i=1}^{n_1} \omega_{t_1}(x_i^{n_1}(t_0^{n_1})+n_1).
\end{multline}
with $\omega_t(x)=e^{-t} t^{x}/x!$.
\end{prop}
\begin{proof}
By the discussion preceding this proposition, we can restrict wlog at $N=n_1$. The measure at time $t$ on $GT_n$ is given by (see (\ref{eqDensityCharlier}))
\begin{equation}
\mu_n(x^n)=\const\, \Delta_n(x^n)^2 \prod_{i=1}^n \omega_t(x_i^n+n).
\end{equation}
Also, recall that (see (\ref{eqMarkovLink}) and (\ref{eqInterlacingDeterminants}))
\begin{equation}
\Lambda^n_{n-1}(x^n,x^{n-1})=\const\, \frac{\Delta_{n-1}(x^{n-1})}{\Delta_n(x^n)} \det(\phi(x_i^{n-1},x_j^n))_{i,j=1}^n
\end{equation}
and that (see (\ref{eqTransProbaCharlier}))
\begin{equation}
P_{n,t}(x^n,y^n)=\frac{\Delta_n(y^n)}{\Delta_n(x^n)}\det(p_t(y_j^n-x_i^n))_{i,j=1}^n.
\end{equation}
Using these identities, $\Delta^n_{m,t}$ defined in (\ref{eqDefDelta}) becomes
\begin{equation}
\begin{aligned}
&\Delta^n_{m,t}(x^n,y^m)=\sum_{z^{m+1},\ldots,z^n}P_{n,t}(x^n,z^n)\Lambda^n_{n-1}(z^n,z^{n-1})\cdots \Lambda^{m+1}_{m}(z^{m+1},y^m)\\
&=\const \frac{\Delta_m(y^m)}{\Delta_n(x^n)} \sum_{z^{m+1},\ldots,z^n} \det(p_t(z_j^n-x_i^n))_{i,j=1}^n \prod_{\ell=m+1}^n\det(\phi(z_i^{\ell-1},z_j^\ell))_{i,j=1}^\ell.
\end{aligned}
\end{equation}
Then, (\ref{eq2.40}) is obtained by multiplying the $\Delta^{n_j}_{n_{j+1},t_{j+1}-t_j}$ of Proposition~\ref{PropJointMeasures} and then reorder all the terms by increasing levels and decreasing times. The notations introduced in the statement avoids to have empty products, like factors $p_{t_{j+1}-t_j}$ when $t_{j+1}=t_j$.
\end{proof}

\subsection{Projection to TASEP}\label{sectProjTASEP}
Already from Proposition~\ref{PropConservationMeasureType} it is obvious that the projection of the Markov chain on $GT_N$ (\ref{eqMConGTn}) onto $W_N$ is still a Markov chain, namely the Charlier process with transition probability $P_{N,t}$. Further, for packed initial conditions, any level $1\leq n\leq N$ evolves as a Charlier process with packed initial condition (i.e., starting from $(-n,-n+1,\ldots,-1)$.

Less obvious is that there are two other projections which are still Markov chains. One is the projection onto $(x_1^1,\ldots,x_1^n)$, the other is the projection onto $(x_1^1,\ldots,x_n^n)$. The first one (that we will discuss here) is TASEP, while the second one is called PushASEP~\cite{BF07}.

\begin{prop}
The projection of the evolution of the Markov chain on $GT_N$ on $(x_1^1,\ldots,x_N^N)$ is still a Markov chain, more precisely, it is the totally asymmetric simple exclusion process with $N$ particles.
\end{prop}
\begin{proof}
It is quite simple to see this fact if we go back to the discrete time version first. This was described at the end of Section~\ref{sectIntertwining}. If we focus only on particles $x_1^n$'s, then they jump to the right with probability $p$ and stay put with probability $1-p$, with the extra condition that particles $x_1^n$ is blocked by $x_1^{n-1}$ for $n\geq 2$. The update is made first for $x_1^1$, then $x_1^2$ and so on. This model is know as TASEP in discrete time with \emph{sequential update} and it is well known that in the continuous time limit, $p\to 0$ with $t\to t/p$, one recovers the continuous time TASEP.
\end{proof}

\newpage Here is a simple but useful observation.
\begin{cor}\label{CorDistrTASEP}
For any choice of $m$ distinct positive integer numbers $n_1,\ldots,n_m$ it holds
\begin{equation}
\Pb\left(\bigcap_{k=1}^m\{x_1^{n_k}\geq s_k\}\right)=\Pb\left(\bigcap_{k=1}^m\{\textrm{No particles at level }n_k\textrm{ is in }(-\infty,s_k)\}\right)
\end{equation}
\end{cor}

\begin{remark}
We have seen that the projection of our Markov chain on $GT_N$ to the $x_1^n$'s has TASEP dynamics. So, any given initial measure $\mu_N$ on $W_N$ induces a measure on the initial condition for TASEP. However, often measures on $\mu_N$ do not lead to ``natural'' initial conditions for TASEP. On the other hand, there are some interesting initial conditions for TASEP, e.g., $x_k=-2k$ for $1\leq k\leq N$, which do not correspond to a probability measure on $\mu_N$, but only to a signed measure. Nevertheless, the algebraic structure goes through and the interlacing structure can still be used to analyze such initial conditions, see e.g.~\cite{BFPS06,BF07,BFS07b,BFS07}.
\end{remark}

\section{$L$-ensembles and determinantal correlations}\label{SectLmeasures}
\subsection{$L$-measures and its correlation functions}
Let us first define determinantal point processes and $L$-ensembles. Our presentation is strongly inspired by~\cite{Bor10,RB04}. Further surveys on determinantal point process are~\cite{Sos00,Sos06,Lyo03,Jo05,Koe04,BKPV05}. Let $\X$ be a discrete space. A \textbf{simple point process} $\eta$ on $\X$ is a probability measure on the set $2^\X$ of all subsets of $\X$. Then, $\eta$ is called \textbf{determinantal} if there exists a function $K:\X\times\X\to \C$ such that for any finite $(x_1,\ldots,x_n)\subset \X$ one has
\begin{equation}\label{eq3.1}
\rho^{(n)}(x_1,\ldots,x_n):=\Pb(X\in 2^\X \, |\, (x_1,\ldots,x_n)\subset X) = \det(K(x_i,x_j))_{i,j=1}^n.
\end{equation}
The function $K$ is called the \textbf{correlation kernel} of $\eta$.
If one thinks at the sites of $\X$ either occupied by a particle or empty, then $\rho^{(n)}(x_1,\ldots,x_n)$ is the probability that each of the sites $x_1,\ldots,x_n$ is occupied by a particle. $\rho^{(n)}$ is also known as $n$-point correlation functions. In the continuous setting, like for the eigenvalues of $N\times N$ GUE random matrices where $\X=\R$, then $\rho^{(n)}(x_1,\ldots,x_n)$ is the probability density\footnote{This, in case the reference measure is Lebesgue. Correlation functions are given with respect to a reference measure. We will here not specify it and use always counting measure for discrete and Lebesgue for the continuous cases. The reference measure in the background is the reason of the factor $b$ in front of the kernel obtained after the change of variable in Lemma~\ref{lemRescalingKernel}} of finding an eigenvalue at each of the $x_1,\ldots,x_n$.

Remark that this does not mean that $\rho^{(n)}$ is normalized to one. For instance, if the point process $\eta$ consists of configurations with exactly $m$ particles, then $\sum_{x\in\X}\rho^{(1)}(x)=m$.

Now, let $L:\X\times\X\to\C$ be a matrix and $\X$ finite (for the moment). For any subset $X\subset \X$ we denote by $L_X$ the symmetric submatrix of $L$ corresponding to $X$, i.e., $L_X=[L(x_i,x_j)]_{x_i,x_j\in X}$. If the determinants of all such submatrices are nonnegative (that is the case for instance if $L$ is positive definite), then one can define a random point process on $\X$ by
\begin{equation}
\Pb(X)=\frac{\det(L_X)}{\det(\Id+L)},\quad X\subset \X.
\end{equation}
This process is called \textbf{$L$-ensemble}.
\begin{thm}[Macchi'75~\cite{Ma75}]
The $L$-ensemble is a determinantal point process with correlation kernel $K=L(\Id+L)^{-1}$.
\end{thm}
Next, consider a (nonempty) subset $\Y$ of $\X$ and a given $L$-ensemble on $\X$. Define a random point process on $\Y$ by considering the intersections of the random point configurations $X\subset \X$ of the $L$-ensemble with $\Y$, provided that these point configurations contain $\Y^c:=\X\setminus \Y$. This new process can be defined by
\begin{equation}
\Pb(Y)=\frac{\det(L_{Y\cup\Y^c})}{\det(\Id_\Y+L)},
\end{equation}
for $Y$ configurations in $\Y$. This process is called \textbf{conditional $L$-ensemble}.
\begin{thm}[Borodin,Rains; Theorem 1.2 of~\cite{RB04}]\label{ThmConditionalLensembles}
The conditional \mbox{$L$-ensemble} is a determinantal point process with correlation kernel
\begin{equation}
 K=\Id_\Y-(\Id_\Y+L)^{-1}\big|_{\Y\times\Y}
\end{equation}
\end{thm}

\begin{rem}
The results extends easily to countable $\X$ by a limiting procedure provided that the normalization constants in the above formulas remain finite (and also to uncountable spaces like $\X=\R$, where of course the standard reference measure becomes Lebesgue instead of counting measure).
\end{rem}

For determinantal point processes, the probability of having a region that is empty (called \textbf{gap probability}) is given by the series expansion of a Fredholm determinant.

\begin{lem}[Gap probability formula]\label{lemGapProbability}
Let $B$ be a (Borel) subset of $\X$. Then, the probability that a random configuration $X=(x_i)_i$ of a determinantal point process with correlation kernel $K$ is empty is equal to
\begin{equation}
\Pb(|X\cap B|=\emptyset)=\det(\Id-K)_{\ell^2(B)}\equiv \sum_{n=0}^\infty \frac{(-1)^n}{n!} \sum_{x_1,\ldots,x_n\in B} \det(K(x_i,x_j))_{i,j=1}^n.
\end{equation}
\end{lem}
\begin{proof}
One just have to write $\Pb(|X\cap B|=\emptyset)$ in terms of correlation functions. For simple point processes we have
\begin{equation}
\begin{aligned}
\Pb(|X\cap B|=\emptyset) &= \E\left(\prod_{i}(1-\Id_{B}(x_i))\right)=\sum_{n\geq 0}(-1)^n\E\left(\sum_{i_1<\ldots<i_n}\prod_{k=1}^n\Id_{B}(x_{i_k})\right)\\
&\stackrel{\textrm{sym}}{=}\sum_{n\geq 0}\frac{(-1)^n}{n!}\E\bigg(\sum_{\begin{subarray}{c}
i_1,\ldots,i_n\\ \textrm{all different}\end{subarray}}\prod_{k=1}^n\Id_{B}(x_{i_k})\bigg)\\
&=\sum_{n\geq 0}\frac{(-1)^n}{n!}\sum_{y_1,\ldots,y_n\in B} \E\bigg(\sum_{\begin{subarray}{c}
i_1,\ldots,i_n\\ \textrm{all different}\end{subarray}}\prod_{k=1}^n\Id_{y_k}(x_{i_k})\bigg)\\
&=\sum_{n\geq 0}\frac{(-1)^n}{n!}\sum_{y_1,\ldots,y_n\in B} \rho^{(n)}(y_1,\ldots,y_n).
\end{aligned}
\end{equation}
Replacing the formula for determinantal point processes of the correlation function the proof is completed.
\end{proof}

Remark that the joint distribution of TASEP particles (see Corollary~\ref{CorDistrTASEP}) can be written as a gap probability. As we will see, the point process in the background is determinantal. It is from this formula that (\ref{eq1.4}) and (\ref{eq1.11}) are obtained after scaling limit. Of course, in the continuous case $\X=\R$, the sum is replaced by the integral with respect to the Lebesgue measure.

In applications one often deals with scaling limits which are affine transformations. Thus, let us shortly write how the kernel of a determinantal point process is changed.
\begin{lem}\label{lemRescalingKernel}
Let $\eta$ be a determinantal point process on $\X$ with kernel $K$ and consider the change of variable $\X \ni x = a+b x'$. Let $\eta'$ be the image of the point process that now live on $\X'=(\X-a)/b$. Then, $\eta'$ is also determinantal with kernel
\begin{equation}
K'(x',y')=b K(a+b x',a+b y').
\end{equation}
\end{lem}
A further important remark is that the kernel is not uniquely determined by the determinantal point process.
\begin{lem}\label{lemConjugate}
Let $\eta$ be a determinantal point process on $\X$ with kernel $K$. For any function $f(x)>0$ for all $x\in \X$, then the kernel
\begin{equation}
\widetilde K(x,y)=\frac{f(x)}{f(y)} K(x,y)
\end{equation}
defines the same determinantal point process $\eta$ as (\ref{eq3.1}) is unchanged.
\end{lem}
In the same way as the moments of random variables do not, in general, determine its law, the correlation functions do not determine always the point process. However, in a lot of applications this holds. A sufficient (an easy to verify) condition implying that the correlation functions determine the point process is that
\begin{equation}
\rho^{(n)}(x_1,\ldots,x_n)\leq n^{2n} c^n\quad a.s.
\end{equation}
for some finite constant $c>0$.

\subsubsection{Example with fixed number of particles}\label{sectFixedNumber}
As first example, let us see how the conditional $L$-ensembles fit in the setting of Eynard-Mehta's theorem~\cite{EM97}. Consider the two-time joint measure of Dyson's Brownian motion with initial measure (\ref{eqGUE}) and transition probability (\ref{eq1.6}), i.e.,
\begin{equation}
\Pb(H(t)\in \dx H_1, H(0)\in \dx H_0)=\const e^{-\frac{1}{2N}\Tr(H_0^2)} e^{-\frac{\Tr(H_1-q H_0)^2}{2N(1-q^2)}} \dx H_0 \dx H_1,
\end{equation}
with $q=e^{-t/2N}$.

Harish-Chandra/Itzykson-Zuber formula~\cite{HC57,IZ80} is the following. Let \mbox{$A=\diag(a_1,\dots,a_N)$} and $B=\diag(b_1,\dots,b_N)$ two diagonal $N\times N$ matrices. Let $\dx\mu$ denote the Haar measure on the unitary group ${\cal U}(N)$. Then,
\begin{equation}\label{eqIZHC1}
\int_{{\cal U}(N)}\dx \mu(U) \exp\left(\Tr(A U B U^*)\right) =
\frac{\det\left(e^{a_i b_j}\right)_{i,j=1}^N}{\Delta_N(a)\Delta_N(b)}\prod_{p=1}^{N-1} p! ,
\end{equation}
where $\Delta_N(a)$ is the Vandermonde determinant of the vector $a=(a_1,\dots,a_N)$.

One writes (for $k=1,2$) $H_k=U_k \Lambda^{(k)} U_k^*$, with $\Lambda^{(k)}=\diag(\lambda_1^{(k)},\ldots,\lambda_N^{(k)})$ is the diagonal matrix with the eigenvalues of $H_k$ and $U_k$ an unitary matrix and uses (\ref{eqIZHC1}) to obtain the joint distribution of the eigenvalues of $H_1$ and $H_0$. By using the identity $\Delta_N(\lambda)= \det(\lambda_i^{j-1})_{i,j=1}^N$, the joint density of the eigenvalues has the form
\begin{equation}
\const \det[\Phi_i(\lambda_j^{(0)})]_{i,j=1}^N \det[\T(x_i^{(0)},x_j^{(1)})]_{i,j=1}^N \det[\Psi_i(\lambda_j^{(1)})]_{i,j=1}^N.
\end{equation}

This is a special case, of the following situation. Consider a random point process on a disjoint union of $m$ (finite) sets $\X^{(1)}\cup \cdots\cup \X^{(m)}$ which lives on $mN$-point configurations with exactly $N$ points in each $\X^{(k)}$, $k=1,\ldots,m$, equipped with the probability measure
\begin{multline}\label{eq3.8}
\Pb\left(\{x_i^{(1)}\}_{1\leq i\leq N}\cap\cdots\cap \{x_i^{(m)}\}_{1\leq i\leq N}\right)=\const \det[\Phi_i(x_j^{(1)})]_{i,j=1}^N \\
\times \prod_{k=1}^{m-1}\det[\T_{k,k+1}(x_i^{(k)},x_j^{(k+1)})]_{i,j=1}^N \det[\Psi_i(x_j^{(m)})]_{i,j=1}^N.
\end{multline}
Now, consider $\X=\{1,\ldots,N\}\cup\X^{(1)}\cup\cdots\cup\X^{(m)}$ and consider the conditional $L$-ensemble on $\X$ with $\Y=\X^{(1)}\cup\cdots\cup\X^{(m)}$, with the matrix $L$ given in block form by
\begin{equation}\label{eqLmatrixDBM}
L=\left(
    \begin{array}{cccccc}
      0 & \Phi^{\mathrm{T}} & 0 & 0 & \cdots & 0 \\
      0 & 0 & -\T_{1,2} & 0 & \cdots & 0 \\
      0 & 0 & 0 & -\T_{2,3} & \cdots & 0 \\
      \vdots & \vdots & \vdots & \vdots & \ddots & \vdots \\
      0 & 0 & 0 & 0 & \cdots & -\T_{m-1,m} \\
      \Psi & 0 & 0 & 0 & \cdots & 0 \\
    \end{array}
  \right),
\end{equation}
where $\Psi=(\Psi_1,\ldots,\Psi_N)$, $\Phi=(\Phi_1,\ldots,\Phi_N)$, and $\Phi^{\mathrm{T}}$ is the transpose of $\Phi$. Then, the conditional $L$-ensemble is the point process distributed according to (\ref{eq3.8}). Indeed, the determinant of a block matrix of type (\ref{eqLmatrixDBM}) is non-zero only if the sizes of all blocks are equal. In that case, the determinant is equal to the product of the determinants of the nonzero blocks (up to a sign). By the choice of $\Y$, $\Y^c=\{1,\ldots,N\}$ so that for any $Y\in\Y$, $\det(L_{Y\cup\Y^c})$ is a determinant of a matrix of the form (\ref{eqLmatrixDBM}) with the $\Psi$ block having $N$ columns. Thus, $\det(L_{Y\cup\Y^c})$ can be non-zero only if $Y\in\Y$ is a configuration with $N$ points in each of the $\X^{(k)}$, $k=1,\ldots,m$, in which case $\det(L_{Y\cup\Y^c})=\const \times (\ref{eq3.8})$.

In the following we use the notation
\begin{equation}
\begin{aligned}
(a*b)(x,y)=\sum_z a(x,z) b(z,y), &\quad (a*c)(x)= \sum_z a(x,z) c(z),\\
(c*a)(x)=\sum_z c(z) a(z,x),&\quad (c*d)=\sum_z c(x) d(x),
\end{aligned}
\end{equation}
for arbitrary functions $a(x,y)$, $b(x,y)$, $c(x)$, and $d(x)$.

An application of Theorem~\ref{ThmConditionalLensembles} (see~\cite{RB04} for details) gives then
\begin{thm}[Eynard-Mehta theorem~\cite{EM97}]\label{thmEynardMehta}
The random point process defined by (\ref{eq3.8}) is determinantal. Its correlation kernel can be written as follows,
\begin{equation}
K(t_1,x_1;t_2,x_2)=-\T_{t_1,t_2}(x_1,x_2)+\sum_{k,\ell=1}^N [G^{-1}]_{k,\ell}  (\T_{t_1,m}*\Psi_k)(x_1) (\Phi_\ell*\T_{1,t_2})(x_2),
\end{equation}
where $t_1,t_2\in\{1,\ldots,m\}$, $G$ is the $N\times N$ matrix
\begin{equation}
G_{i,j}=\Phi_i *\T_{1,2}*\cdots *\T_{m-1,m}*\Psi_j,
\end{equation}
and
\begin{equation}
\T_{i,j}(x,y)=\left\{
\begin{array}{ll}
(\T_{i,i+1}*\cdots * \T_{j-1,j})(x,y),& i<j,\\
0, & i\geq j.
\end{array}
\right.
\end{equation}
\end{thm}

\subsubsection{Example with increasing number of particles}
The next example is motivated by the Markov chain in the interlacing particle configurations. Consider the measure (\ref{eq2.20}) with $\mu_N(x^N)$ of the form $\Delta_N(x^N)^2 \prod_{k=1}^N \omega(x_k^N)$ (as it is the case in (\ref{eqDensityCharlier})). Then, by (\ref{eqMarkovLink}) and Lemma~\ref{LemInterlacing} we obtain a measure of the form
\begin{equation}\label{eq3.12}
\const \prod_{n=1}^{N}\det(\phi_n(x_i^{n-1},x_j^n))_{i,j=1}^n \det(\Psi_i(x_j^N))_{i,j=1}^N,
\end{equation}
with $x_n^{n-1}\equiv \virt$, there are some given functions
\begin{equation}\label{eq3.13}
\begin{aligned}
\phi_n(\cdot,\cdot):\X_{n-1}\times \X_n\to \C,&\quad n=2,\ldots,N,\\
\phi_n(\virt,\cdot):\X_n\to\C,&\quad n=1,\ldots,N,\\
\Psi_i(\cdot):\X_N\to\C,&\quad i=1,\ldots,N,
\end{aligned}
\end{equation}
and where $(x_1^n,\ldots,x_n^n)$ is a $n$-point configuration in a space $\X_n$. Take \mbox{$\X=\{1,\ldots,N\}\times\cup\Y$} with $\Y=\X_1\cup\ldots\cup\X_N$, with $\X_n=\Z$ is the space where the $n$ variables at level $n$ live. Consider the conditional $L$-ensemble with matrix $L$ given by
\begin{equation}
L=\left(
    \begin{array}{cccccc}
      0 & E_0 & E_1 & E_2 & \cdots & E_{N-1} \\
      0 & 0 & -\phi_{1,2} & 0 & \cdots & 0 \\
      0 & 0 & 0 & -\phi_{2,3} & \cdots & 0 \\
      \vdots & \vdots & \vdots & \vdots & \ddots & \vdots \\
      0 & 0 & 0 & 0 & \cdots & -\phi_{N-1,N} \\
      \Psi & 0 & 0 & 0 & \cdots & 0 \\
    \end{array}
  \right),
\end{equation}
where
\begin{equation}
\begin{aligned}
\left[\Psi\right]_{x,i}&=\Psi_i(x),\quad x\in \X_N,i\in\{1,\ldots,N\},\\
[E_n]_{i,x}&=\left\{
\begin{array}{ll}
\phi_{n+1}(\virt,x),&\quad i=n+1,x\in\X_{(n+1)},\\
0,&\quad \textrm{otherwise},
\end{array}
\right.\\
[\phi_{n,n+1}]_{x,y}&=\phi_{n+1}(x,y),\quad x\in\X_{(n)},y\in\X_{(n+1)}.
\end{aligned}
\end{equation}
For example, when $N=2$ and $Y=(x_1^1)\cup(x_1^2,x_2^2)$, then
\begin{equation}
L_{Y\cup\Y^c} =
\left(
  \begin{array}{ccccc}
    0 & 0 & \phi_1(\virt,x_1^1) & 0 & 0 \\
    0 & 0 & 0 & \phi_2(\virt,x_1^2) & \phi_2(\virt,x_2^2) \\
    0 & 0 & 0 & -\phi_2(x_1^1,x_1^2) & -\phi_2(x_1^1,x_2^2) \\
    \Psi_1(x_1^2) & \Psi_2(x_1^2) & 0 & 0 & 0 \\
    \Psi_1(x_2^2) & \Psi_2(x_2^2) & 0 & 0 & 0 \\
  \end{array}
\right)
\end{equation}

We want to see that for a configuration $Y\in\Y$, then \mbox{$\det(L_{Y\cup\Y^c})=\const\times(\ref{eq3.12})$} provided that $Y$ has exactly $n$ points in $\X_n$, and otherwise $\det(L_{Y\cup\Y^c})=0$. To see that this is the case, first notice that if $Y$ has $c$ points in $\X_n$, then the matrix has $c$ columns from $E_n$ filled with $0$ except for the ($n+1$)th row. So, for $c>n$, the matrix does not has full rank and its determinant is zero. Thus, the number of points in $\X_n$ is at most $n$. Further, if there are strictly less than $N$ points in $\X_N$, the matrix is also not full rank because of the columns coming from $\Psi$. Given this, if the number of points in $\X_{N-1}$ is strictly less than $N-1$, then looking at the columns from $E_{N-1}$ one sees that the matrix is not full rank either. Similarly for $n=N-2,N-3,\ldots,1$.

By using Theorem~\ref{ThmConditionalLensembles} we get the following (see Lemma 3.4 of~\cite{BFPS06} for details):
\begin{thm}[Borodin, Ferrari, Pr\"ahofer, Sasamoto; Lemma 3.4 of~\cite{BFPS06}]\label{ThmMinors}
The random point process on $\X$ defined by (\ref{eq3.12}) is determinantal. Its correlation kernel can be written as follows.
Define the functions
\begin{equation}
\phi_{n_1,n_2}(x,y)=\left\{
\begin{array}{ll}
(\phi_{n_1+1}* \cdots * \phi_{n_2})(x,y),&\quad n_1<n_2,\\
0,&\quad n_1\geq n_2.
\end{array}
\right.
\end{equation}
Then,
\begin{equation}
\begin{aligned}
K(n_1,x_1;n_2,x_2)&=-\phi_{n_1,n_2}(x_1,x_2)\\
&+ \sum_{\ell=1}^{n_2}\sum_{k=1}^N [G^{-1}]_{k,\ell} (\phi_{n_1,N}*\Psi_k)(x_1) (\phi_{\ell}*\phi_{\ell,n_2})(\virt,x_2),
\end{aligned}
\end{equation}
where $G$ is the $N\times N$ matrix defined by $[G]_{i,j}=(\phi_i * \phi_{i,N} * \Psi_j)(\virt)$.
\end{thm}

\begin{remark}
To have a manageable form of the kernel and take asymptotics, one usually tries to find a change of basis in (\ref{eq3.8}), resp.\ (\ref{eq3.12}), such that the matrix $G$ to be inverted becomes the identity matrix. In the classical examples this is possible by using orthogonal polynomials. Generically, one looks for biorthogonal ensembles~\cite{Bor98} (see for instance Lemma 3.4 of~\cite{BFPS06}). See also Section~\ref{sectApplication} below for an application.
\end{remark}

\subsubsection{A generalization}
Consider $c(1),\ldots,c(N)$ be arbitrary nonnegative integers and let
\begin{equation}
t_0^N\leq \cdots\leq t_{c(N)}^N=t_0^{N-1}\leq \cdots\leq t_{c(N-1)}^{N-1}=t_0^{N-2}\leq \cdots\leq t_{c(2)}^2=t_0^1\leq \cdots\leq t_{c(1)}^1
\end{equation}
be real numbers (that in our case are the observation times of the state of our Markov chain on $GT_N$). Let $\phi_n$ and $\Psi$ be as in (\ref{eq3.13}) and
\begin{equation}
\T_{t_a^n,t_{a-1}^n}(\cdot,\cdot):\X_n\times\X_n\to \C,\quad n=1,\ldots,N,\quad a=1,\ldots,c(n)
\end{equation}
be arbitrary functions.

Let our configurations live in the space \mbox{$\Y=\X^{(1)}\cup \cdots\cup \X^{(N)}$}, with \mbox{$\X^{(n)}=\X_0^{(n)}\cup\cdots\cup\X_{c(n)}^{(n)}$}, where each of the $\X_a^{(n)}$ is a copy of the space $\X_n$ where the variables live at time $t_a^n$. Then, consider the point process whose point configurations $Y\in \Y$ have weight zero unless it has exactly $n$ points in each copy of $\X_n$, $n=1,\ldots,N$. In the latter case, we denote by $x_k^n(t_a^n)$ the points of $Y$ in the $a$-th copy of $\X_n$, for $k=1,\ldots,n$, and we assign a measure of $Y$ given by
\begin{multline}\label{eqPushASEPmeasure}
\const\, \prod_{n=1}^{N}\bigg[\det[\phi_n(x_i^{n-1}(t_0^{n-1}),x_j^n(t_{c(n)}^n))]_{i,j=1}^n\\
\times \prod_{a=1}^{c(n)}\det[\T_{t_a^n,t_{a-1}^n}(x_j^n(t_a^n),x_i^n(t_{a-1}^n))]_{i,j=1}^n\bigg] \det[\Psi_{i}(x_j^N(t_0^{N}))]_{i,j=1}^N,
\end{multline}
where again $x_n^{n-1}(\cdot)=\virt$ for all $n=1,\ldots,N$. Remark that (\ref{eq2.40}) is a special case of such a measure.

A measure of the form (\ref{eqPushASEPmeasure}) is determinantal. To describe the kernel we need some notations. For any $n=1,\ldots,N$ and two time moments $t_a^n>t_b^n$, we define
\begin{equation}
\T_{t_a^n,t_b^n}=\T_{t_a^n,t_{a-1}^n} * \T_{t_{a-1}^n,t_{a-2}^n} * \cdots * \T_{t_{b+1}^n,t_b^n},\quad \textrm{and}\quad \T^n=\T_{t^n_{c(n)},t_0^n}.
\end{equation}
Further, for any two time moments $t_{a_1}^{n_1}\geq t_{a_2}^{n_2}$ with $(a_1,n_1)\neq (a_2,n_2)$, we denote by
\begin{equation}
\phi^{(t_{a_1}^{n_1},t_{a_2}^{n_2})}= \T_{t_{a_1}^{n_1},t_{0}^{n_1}}*\phi_{n_1+1}*\T^{n_1+1}*\cdots * \phi_{n_2} * \T_{t_{c(n_2)}^{n_2},t_{a_2}^{n_2}}
\end{equation}
the convolution over all the transitions between them. Is there are no such transitions, i.e., if we do not have $(t_{a_1}^{n_1},n_1)\prec (t_{a_2}^{n_2},n_2)$, then we set \mbox{$\phi^{(t_{a_1}^{n_1},t_{a_2}^{n_2})}=0$}. Finally, define the $N\times N$ matrix $G$ by
\begin{equation}
G_{i,j}=(\phi_i * \T^i * \cdots * \phi_N * \T^N * \Psi_j)(\virt),\quad 1\leq i,j\leq N
\end{equation}
and set
\begin{equation}
\Psi^{t_a^n}_j=\phi^{(t_a^n,t_0^N)} * \Psi_j,\quad 1\leq j\leq N.
\end{equation}
Then, by applying Theorem~\ref{ThmConditionalLensembles} one proves the following.
\begin{thm}[Borodin, Ferrari; Theorem 4.2 of~\cite{BF07}]\label{thmPushASEPKernel}
The random point process on $\Y$ defined by (\ref{eqPushASEPmeasure}) is determinantal. Its correlation kernel can be written as
\begin{equation}
\begin{aligned}
K(t_{a_1}^{n_1},x_1;t_{a_2}^{n_2},x_2)=&-\phi^{(t_{a_1}^{n_1},t_{a_2}^{n_2})}(x_2,x_1)\\
&+\sum_{k=1}^{N} \sum_{\ell=1}^{n_2} [G^{-1}]_{k,\ell} \Psi_k^{t_{a_1}^{n_1}}(x_1) (\phi_\ell * \phi^{(t^\ell_{c(\ell)},t^{n_2}_{a_2})})(\virt,x_2).
\end{aligned}
\end{equation}
In the case when the matrix $G$ is upper triangular, there is a simpler way to write the kernel.
Set
\begin{equation}\label{eq4.37}
\Phi^{t_a^n}_k(x)=\sum_{\ell=1}^n [G^{-1}]_{k,\ell} \big(\phi_\ell * \phi^{(t^\ell_{c(\ell)},t^n_a)}\big)(\virt,x)
\end{equation}
for all $n=1,\ldots,N_1$ and $k=1,\ldots,n$. Then, $\big\{\Phi^{t^n_a}_{k}\big\}_{k=1,\ldots,n}$ is the unique basis of the linear span of
\begin{equation}\label{eq2.16}
\Big\{(\phi_1* \phi^{(t^1_{c(1)},t^{n}_a)})(\virt,x), \ldots, (\phi_{n}*\phi^{(t^{n}_{c(n)},t^{n}_a)})(\virt,x) \Big\}
\end{equation}
that is biorthogonal to $\{\Psi^{t^n_a}_{k}\}$, i.e., satisfying
\begin{equation}
\Phi^{t^n_a}_i *\Psi^{t^n_a}_j = \delta_{i,j},\quad i,j=1,\ldots,n.
\end{equation}
The correlation kernel can then be written as
\begin{equation}
K(t^{n_1}_{a_1},x_1; t^{n_2}_{a_2},x_2)= -\phi^{(t^{n_1}_{a_1},t^{n_2}_{a_2})}(x_1,x_2) + \sum_{k=1}^{n_2} \Psi^{t^{n_1}_{a_1}}_{k}(x_1) \Phi^{t^{n_2}_{a_2}}_{k}(x_2).
\end{equation}
\end{thm}

\subsection{Application to the measure of Proposition~\ref{PropMeasureContTime}}\label{sectApplication}
Consider the measure (\ref{eq2.40}) obtained by starting the Markov chain on $GT_N$ from \textbf{packed initial condition}, i.e., with $x_k^n(0)=-n-1+k$ for \mbox{$1\leq k \leq n \leq N$}, see Figure~\ref{Figure2dDynamics} for an illustration.
\begin{figure}[t!]
\begin{center}
\psfrag{x11}[cl]{$x_1^1$}
\psfrag{x12}[cl]{$x_1^2$}
\psfrag{x13}[cl]{$x_1^3$}
\psfrag{x14}[cl]{$x_1^4$}
\psfrag{x15}[cl]{$x_1^5$}
\psfrag{x22}[cl]{$x_2^2$}
\psfrag{x23}[cl]{$x_2^3$}
\psfrag{x24}[cl]{$x_2^4$}
\psfrag{x25}[cl]{$x_2^5$}
\psfrag{x33}[cl]{$x_3^3$}
\psfrag{x34}[cl]{$x_3^4$}
\psfrag{x35}[cl]{$x_3^5$}
\psfrag{x44}[cl]{$x_4^4$}
\psfrag{x45}[cl]{$x_4^5$}
\psfrag{x55}[cl]{$x_5^5$}
\includegraphics[width=\textwidth]{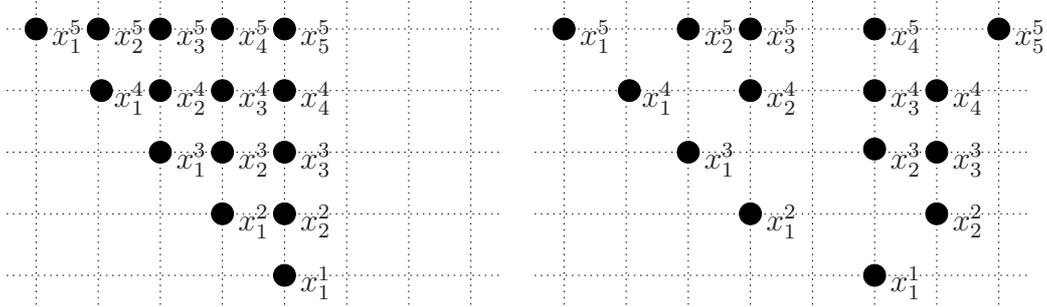}
\caption{(Left) Illustration of the initial conditions for the particles system. (Right) A configuration obtained from the initial conditions. For a Java animation of the model see~\cite{FerAKPZ}.}
\label{Figure2dDynamics}
\end{center}
\end{figure}
Then, since the Vandermonde determinant is a determinant, (\ref{eq2.40}) is a measure of the form (\ref{eqPushASEPmeasure}).
The goal of this section is to determine an explicit and ``asymptotic friendly'' formula for the correlation kernel. By noticing that the probability distribution of a given number of TASEP particles is nothing else that a gap probability, it follows from Lemma~\ref{lemGapProbability} that their joint distribution is a Fredholm determinant of the correlation kernel. Large time asymptotics are then obtained by analyzing the kernel under certain scaling limit.

\subsubsection{Correlation kernel}
To obtain the correlation kernel, first a result on the measure at time $t=0$ on $GT_N$. This is obtained by setting $\mu_N$ to be the delta-measure on \mbox{$x^N=(-N,-N+1,\ldots,-1)$}. The interlacing condition fixes then all the particles at lower levels and the measure on $GT_N$ is given by (\ref{eq2.28a}).
\begin{lem}\label{lemPackedIC}
Consider the probability measure on $W_N$ given by
\begin{equation}
\mu_N(x^N)=\const\, \Delta_N(x^N) \det(\Psi_i(x_j^N))_{i,j=1}^N
\end{equation}
where
\begin{equation}\label{eq3.35}
\Psi_i(x)=\frac{1}{2\pi\I}\oint_{\Gamma_0}\dx z (1-z)^{N-i}z^{x+i-1},\quad 1\leq i\leq N.
\end{equation}
Then $\mu_N(x^N)=\delta_{x^N,(-N,\ldots,-1)}$.
\end{lem}
\begin{proof}
First notice that $\textrm{span}(\Psi_1,\ldots,\Psi_N)$ is exactly the space of all functions on $\Z$ supported by $\{-N,\ldots,-1\}$. This is due to (1)  $\Psi_i$ is supported on $\{-N,\ldots,-i\}$, (2) they are linearly independent since $\Psi_i(-i)=1\neq 0$, and (3) the mentioned space has dimension $N$. Thus, $\det(\Psi_i(x_j^N))_{i,j=1}^N=0$ if at least one of the $x_j^N\not\in\{-N,\ldots,-1\}$. By antisymmetry of the determinant $\det(\Psi_i(x_j^N))_{i,j=1}^N=0$ if two of the $x_j^N$ are equal. Thus the only configuration in $W_N$ such that $\det(\Psi_i(x_j^N))_{i,j=1}^N\neq 0$ is $x^N=(-N,\ldots,-1)$.
\end{proof}

The reason of the choice of the particular form of $\Psi$ in Lemma~\ref{lemPackedIC} becomes clear when one starts doing the computations and apply $\T$ and the $\phi_n$'s. It is easy to see that the measure of Proposition~\ref{PropMeasureContTime} is a special case of the measure (\ref{eqPushASEPmeasure}). Indeed, we need to set $t_0^N=0$, $\Psi_i(x)$ as in (\ref{eq3.35}), and
\begin{equation}
\begin{aligned}
\T_{t_i,t_j}(x,y)&=\frac{1}{2\pi\I}\oint_{\Gamma_0}\dx z e^{(t_j-t_i)/z} z^{x-y-1}=\frac{1}{2\pi\I}\oint_{\Gamma_0}\dx w \frac{e^{(t_j-t_i)w}}{w^{x-y+1}},\\
\phi_n(x,y)&=\Id_{[y\geq x]}=\frac{1}{2\pi\I}\oint_{\Gamma_0}\dx w \frac{(1-w)^{-1}}{x^{x-y+1}},\quad \phi_n(\virt,y)=1.
\end{aligned}
\end{equation}

With these ingredients we can apply Theorem~\ref{thmPushASEPKernel} and obtain the kernel. This needs some computations, which can be found in the proof of Proposition 3.1 of~\cite{BF07} and its specialisation to the so-called ``step initial condition'' in Section 3.2 of~\cite{BF07}. The result is the following.

\begin{thm}[Proposition 3.4 of~\cite{BF07}]\label{thmspacelikekernel}
The point process issued by the continuous Markov chain on $GT_N$ of Section~\ref{sectIntertwining}, i.e., generating the measure of Proposition~\ref{PropMeasureContTime}, is determinantal along ``space-like paths'', i.e., if we look at decreasing levels by increasing times. That is, for any $m\in \N$, pick $m$ (distinct) triples
\begin{equation}
(n_i,t_i,x_j)\in\N\times\R_+\times\Z
\end{equation}
such that $n_1\leq n_2\leq \ldots\leq n_m$ and $t_1\geq t_2\geq \ldots\geq t_m$. Then,
\begin{multline}
\Pb(\{\textrm{For each }j=1,\ldots,m\textrm{ there is a }k_j, 1\leq k_j\leq n_j\\
\textrm{ such that }x_{k_j}^{n_j}(t_j)=x_j\})=\det[K(n_i,t_i,x_i;n_j,t_j,x_j)]_{i,j=1}^m,
\end{multline}
where the correlation kernel can be written as\footnote{Here we write $(n_1,t_1)$ instead of $t_{a_1}^{n_1}$ to make more explicit the dependence on the two entries, time and level.}
\begin{equation}
\begin{aligned}
&K(n_1,t_1,x_1;n_2,t_2,x_2)\\
&=-\frac{1}{2\pi\I}\oint_{\Gamma_{0}}\dx w \frac{1}{w^{x_1-x_2+1}}\left(\frac{w}{1-w}\right)^{n_2-n_1} e^{(t_1-t_2)w}\Id_{[(n_1,t_1)\prec (n_2,t_2)]}\\
&+\frac{1}{(2\pi\I)^2}\oint_{\Gamma_1}\dx z\oint_{\Gamma_{0}}\dx w \frac{e^{t_1 w}(1-w)^{n_1}}{w^{x_1+n_1+1}}\frac{z^{x_2+n_2}}{e^{t_2 z}(1-z)^{n_2}}\frac{1}{w-z}.
\end{aligned}
\end{equation}
\end{thm}

Using Corollary~\ref{CorDistrTASEP} and Lemma~\ref{lemGapProbability} the joint distribution of TASEP particles are distributed as follows.
\begin{cor}\label{corTASEP}
Consider a system of TASEP particles starting with step initial condition, i.e., with $x^n_1(0)=-n$ for $n\geq 1$. Denote by $x_1^n(t)$ the position of particle with index $n$ at time $t$. Then, provided $(n_1,t_2)\prec\cdots\prec(n_m,t_m)$, the joint distribution of particle positions is given by the Fredholm determinant
\begin{equation}
\Pb\left(\bigcap_{k=1}^m\{x_1^{n_k}(t_k)\geq s_k\}\right)=\det(\Id-P_s K P_s)_{\ell^2(\{(n_1,t_1),\ldots,(n_m,t_m)\}\times\Z)}
\end{equation}
with $P_s((n_k,t_k))(x)=\Id_{[x<s_k]}$. Explicitly, the above Fredholm determinant can be written as
\begin{equation}\label{eq3.41}
\sum_{n\geq 0}\frac{(-1)^n}{n!}\sum_{\ell_1=1}^m\cdots \sum_{\ell_n=1}^m\sum_{x_1<s_{\ell_1}}\cdots\sum_{x_n<s_{\ell_n}} \det(K(n_{\ell_i},t_{\ell_i},x_i;n_{\ell_j},t_{\ell_j},x_j))_{i,j=1}^n,
\end{equation}
with the kernel $K$ as in Theorem~\ref{thmspacelikekernel}.
\end{cor}

\subsubsection{Diffusion scaling limit}
Now we consider the diffusion scaling limit for a fixed number of particles. Let us define the rescaled random variables
\begin{equation}
\xi_k^n(\tau):=\lim_{\e\to 0}\e\left(x_k^n(\tfrac12\tau \e^{-2})-\tfrac12 \tau \e^{-2}\right)
\end{equation}
Then, the correlation function of the $\xi_k^n$'s are, along space-like paths, still determinantal with kernel given by
\begin{equation}\label{eq3.43}
\begin{aligned}
K(n_1,\tau_1,\xi_1;n_2,\tau_2,\xi_2)=&-\frac{2}{2\pi\I}\int_{\I\R+\delta}\dx w \frac{e^{(\tau_1-\tau_2)w^2-2(\xi_1-\xi_2)w}}{w^{n_2-n_1}}\Id_{[(n_1,\tau_1)\prec (n_2,\tau_2)]}\\
&+\frac{2}{(2\pi\I)^2}\oint_{|z|=\delta/2}\dx z\int_{\I\R+\delta}\dx w \frac{e^{\tau_1 w^2-2\xi_1 w}w^{n_1}}{e^{\tau_2 z^2-2\xi_2 z}z^{n_2}}\frac{1}{w-z}
\end{aligned}
\end{equation}
where $\delta>0$ is arbitrary.

This is obtained by first defining the rescaled kernel according to Lemma~\ref{lemRescalingKernel} and then doing some asymptotic analysis on the kernel taking the $\e\to 0$ limit (the counting measure becomes the Lebesgue measure).

Since, for finite $n$ the measure if actually a finite determinant (of a matrix of size $n(n+1)/2$) of the correlation kernel, then from the convergence of the kernel follows also the convergence of the measure. In particular, the measure of Proposition~\ref{PropMeasureContTime} converges to the one with delta initial measure and the transition kernel $\T$ becomes the heat kernel.

\subsubsection{Large time limit at the edge}\label{sectAsymptAnalysis}
We can also let the particle level that we focus go to infinity linearly with time. The macroscopic behavior is described in detail in Section 3.1 of~\cite{BF08}. In the bulk the correlation kernel becomes the (extended) sine kernel, one has Gaussian fluctuations and sees a Gaussian Free Field, see Theorem 1.2 and Theorem 1.3 in~\cite{BF08}. Here we want to discuss the case of the edge, i.e., describe the rescaled particle process around the TASEP particles.

For simplicity, we describe the system at a fixed time. A statement along space-like paths can be found in Section 2.3 of~\cite{BF07}.
For a fixed $\alpha\in (0,1)$. TASEP particles with index $n$ close to $\alpha t$ are around position $(1-2\sqrt{\alpha})t$. More precisely, consider the scaling  of level and position at time $t$ given by
\begin{equation}
\begin{aligned}
n(u)&=\alpha t + 2 u t^{2/3},\\
x(u)&=(1-2\sqrt{\alpha})t-\frac{2 u}{\sqrt{\alpha}} t^{2/3}+\frac{u^2}{\alpha^{3/2}} t^{1/3}.
\end{aligned}
\end{equation}
Accordingly, define the rescaled TASEP particle process given by
\begin{equation}
X_{t}(u)=\frac{x_1^{n(u)}(t)-x(u)}{-t^{1/3}}.
\end{equation}
By Corollary~\ref{corTASEP} the joint distributions of the rescaled process is also given by a Fredholm determinant (just do the change of variables).

The analysis of the $t\to\infty$ limit can be made as follows (for a sketch of it see Section 5.2 of~\cite{BF07}, where the scaling holds also for generic space-like paths).\\[6pt]
(1) Define the rescaled kernel as in Lemma~\ref{lemRescalingKernel}, i.e., let
\begin{equation}\label{eq3.50}
K^{\rm resc}_t(u_1,s_1;u_2,s_2):=t^{1/3} K(n(u_1),t,x(u_1)-s_1 t^{1/3};n(u_2),t,x(u_2)-s_2 t^{1/3})
\end{equation}
where $K$ is as in Theorem~\ref{thmspacelikekernel}.\\[6pt]
(2) Do the steep descent analysis of the kernel under this rescaling (see e.g. Section 6.1 of~\cite{BF08} for a description of the single integral case, which can be easily adapted also to double integrals). The leading term in the double integral will come from a region around a double critical point $z_c$. In this case, one has $z_c=1-\sqrt{\alpha}$. After controlling the error terms of the integrals away from the critical points, one does the change of variable \mbox{$z=z_c+(\kappa t)^{-1/3} Z$} and $w=z_c+(\kappa t)^{-1/3} W$ with $\kappa=1/(\sqrt{\alpha}(1-\sqrt{\alpha}))$ and use Taylor approximation. Further controls on the error terms in the Taylor expansion leads to, for $n_1\geq n_2$,
\begin{multline}
\lim_{t\to\infty}K^{\rm resc}_t(u_1,s_1;u_2,s_2)\\
\equiv \frac{S_v^{-1}}{(2\pi\I)^2} \int \dx W \int \dx Z \frac{1}{Z-W} \frac{e^{\frac13 Z^3+ u_2 Z^2/S_h -Z(s_2/S_v-u_2^2/S_h^2)}}{e^{\frac13 W^3+ u_1 W^2/S_h -W(s_1/S_v-u_1^2/S_h^2)}}
\end{multline}
where
\begin{equation}
S_v=\frac{(1-\sqrt{\alpha})^{2/3}}{\alpha^{1/6}},\quad S_h=\alpha^{2/3} (1-\sqrt{\alpha})^{1/3}
\end{equation}
and $\equiv$ means that the equality holds up to a conjugation factor (see Lemma~\ref{lemConjugate}). The integration contours for $W$ and $Z$ can be chosen to be $-\delta+\I\R$ and $\delta+\I\R$ respectively (for any $\delta>0$ and oriented with increasing imaginary part).\\[6pt]
(3) Replacing $1/(Z-W)$ by $\int_{\R_+}\dx\lambda e^{-\lambda (Z-W)}$ and using the integral representation of the Airy functions one has the equality (see e.g. Appendix A of~\cite{BFP09})
\begin{equation}
(\ref{eq3.50})\equiv S_v^{-1} K_2(s_1/S_v,u_1/S_h;s_2/S_v,u_2/S_h),
\end{equation}
where $K_2$ is the extended Airy kernel given in (\ref{eq1.12}).\\[6pt]
(4) Finally, to see that the joint distributions converges, one has to show the convergence of the Fredholm determinants. For this is enough to get some uniform in $t$ estimates in the decay of the kernel for large $s_1,s_2$ and then apply dominated convergence (using also Hadamard's bound that says that the determinant of a $n\times n$ matrix with entries of absolute value not exceeding $1$ is bounded by $n^{n/2}$).
\smallskip

With the procedure described above one obtains
\begin{equation}\label{eq3.49}
\lim_{t\to\infty} X_{t}(u)=S_v{\cal A}_2(u /S_h),
\end{equation}
where ${\cal A}_2$ is the Airy$_2$ process defined by (\ref{eq1.11}) (in the sense of finite dimensional distributions).

\newpage
\section{Random matrices}\label{SectRandMatr}
In this section we go back to random matrix diffusions and will see the similarities with interacting particles above.

\subsection{Random matrix diffusion}
Instead of considering stationary Dyson's Brownian motion, to make the connection more straightforward, we replace the Ornstein-Uhlenbeck processes by Brownian motions starting from $0$. The two models are the same after an appropriate change of scale in space-time.

Let $H(t)$ be an $N\times N$ Hermitian matrix defined by
\begin{equation}
H_{i,j}(t)=
\begin{cases}
\frac{1}{\sqrt{2}}\, b_{i,i}(t), & \textrm{if }1\leq i \leq N, \\
\frac12(b_{i,j}(t)+\I \, \tilde b_{i,j}(t)), & \textrm{if }1\leq i < j \leq N, \\
\frac12(b_{i,j}(t)-\I \, \tilde b_{i,j}(t)), & \textrm{if }1\leq j < i \leq N,
\end{cases}
\end{equation}
where $b_{i,j}(t)$ and $\tilde b_{i,j}(t)$ are independent standard Brownian motions.
The measure on the $N\times N$ matrix at time $t$ is then given by
\begin{equation}
\const\, \exp\left(-\frac{\Tr(H^2)}{t}\right) \dx H.
\end{equation}

For $0<t_1<t_2<\ldots<t_m$, the joint distribution of $H_1=H(t_1)$, $H_2=H(t_2),\ldots, H_m=H(t_m)$ is given by
\begin{equation}\label{eq32}
{\rm const}\times \exp\left(-\frac{\Tr(H_1^2)}{t_1}\right) \prod_{k=1}^{m-1}\exp\left(-\frac{\Tr((H_{k+1}-H_k)^2)}{t_{k+1}-t_k}\right)\dx H_1\,\cdots\, \dx H_m.
\end{equation}
The measure on eigenvalues the Harish-Chandra/Itzykson-Zuber formula~\cite{HC57,IZ80} (\ref{eqIZHC1}). The result is the following.
\begin{lem}\label{LemGUEFixedn}
Denote by $\lambda(t)=(\lambda_1(t),\dots,\lambda_N(t))$ the eigenvalues of $H(t)$. Their joint distribution at $0<t_1<t_2<\ldots<t_m$ is given by
\begin{multline}
{\rm const}\times \Delta_N(\lambda(t_1)) \prod_{k=1}^{m-1}\det\left(e^{-(\lambda_{i}(t_{k})-\lambda_{j}(t_{k+1}))^2/(t_{k+1}-t_k)}\right)_{i,j=1}^N \Delta_N(\lambda(t_m)) \\
\times\prod_{i=1}^N e^{-(\lambda_{i}(t_1))^2 /t_1} \, \dx \lambda_{i}(t_1)\cdots\dx \lambda_{i}(t_m).
\end{multline}
\end{lem}
This measure is a particular case of the setting discussed in Section~\ref{sectFixedNumber} and thus one can apply Eynard-Mehta theorem (Theorem~\ref{thmEynardMehta}) to determine its correlation kernel. The result is
\begin{equation}
\begin{aligned}
K(t_1,x_1;t_2,x_2)=&-\frac{2}{2\pi\I}\int_{\I\R+\delta}\dx w\, e^{(t_1-t_2)w^2-2(x_1-x_2)w}\Id_{[t_1<t_2]}\\
&+\frac{2}{(2\pi\I)^2}\oint_{|z|=\delta/2}\dx z\int_{\I\R+\delta}\dx w\, \frac{e^{t_1 w^2-2x_1 w}}{e^{t_2 z^2-2x_2 z}}\frac{1}{w-z}
\end{aligned}
\end{equation}
where $\delta>0$ is arbitrary. Notice that this kernel is a special case of the kernel (\ref{eq3.43}) obtained in the diffusion scaling limit of the interlacing particle system.

\subsection{GUE minor process}
Instead of considering the evolution of the eigenvalues, one can also consider the eigenvalues of the $N$ principal minors of the matrix $H(t)$.

Denote by $\lambda_k^m$ the $k$th smallest eigenvalue of the principal submatrix obtained from the first $m$ rows and columns of a GUE matrix. In our context, these principal submatrices are usually referred to as minors, and not (as otherwise customary) their determinants. The result is well known: \emph{given the eigenvalues $\lambda_k^N$, $1\leq k\leq N$, of the $N\times N$ matrix}, the GUE minors' eigenvalues are uniformly distributed on the set
\begin{equation}\label{eq18}
{\cal D}^{(N)}=\{\lambda_k^m\in \R, 1\leq k \leq m\leq N|\, \lambda_k^{m+1} \leq \lambda_k^m\leq \lambda_{k+1}^{m+1}, 1\leq k \leq m\leq N\}
\end{equation}
as shown in~\cite{Bar01}. Note that this is the continuous analogue of $GT_N$, see Figure~\ref{FigGUEminors} for an illustration.
\begin{figure}
\begin{center}
\psfrag{l11}[cc]{$\lambda_1^1$}
\psfrag{l12}[cc]{$\lambda_1^2$}
\psfrag{l22}[cc]{$\lambda_2^2$}
\psfrag{l13}[cc]{$\lambda_1^3$}
\psfrag{l23}[cc]{$\lambda_2^3$}
\psfrag{l33}[cc]{$\lambda_3^3$}
\psfrag{l14}[cc]{$\lambda_1^4$}
\psfrag{l24}[cc]{$\lambda_2^4$}
\psfrag{l34}[cc]{$\lambda_3^4$}
\psfrag{l44}[cc]{$\lambda_4^4$}
\psfrag{n}[cc]{$n$}
\psfrag{R}[cc]{$\R$}
\psfrag{1}[cc]{$1$}
\psfrag{2}[cc]{$2$}
\psfrag{3}[cc]{$3$}
\psfrag{4}[cc]{$4$}
\includegraphics[height=4cm]{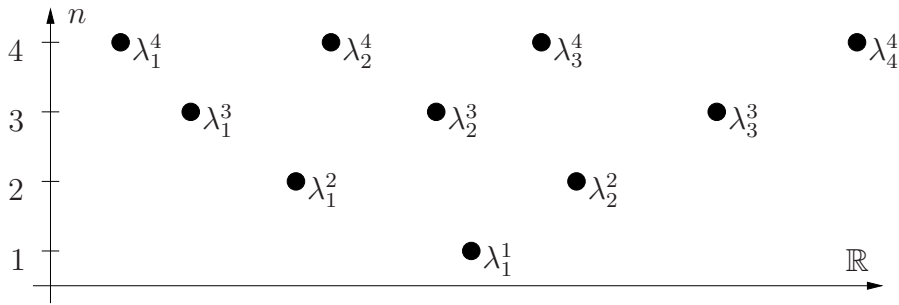}
\caption{Interlacing structure of the GUE minors' eigenvalues.}
\label{FigGUEminors}
\end{center}
\end{figure}

It is proven in~\cite{JN06} that the correlation functions of these eigenvalues are determinantal with correlation kernel
\begin{multline}
K(x_1,n_1;x_2,n_2)=-\frac{2}{2\pi\I}\int_{\I \R+\delta}\dx w\, \frac{e^{-2(x_1-x_2)w}}{w^{n_2-n_1}}\Id_{[n_1<n_2]}\\
+\frac{2}{(2\pi\I)^2}\oint_{|z|=\delta/2} \dx z \int_{\I \R+\delta}\dx w \,\frac{e^{w^2-2w x_1}}{e^{z^2-2z x_2}}\frac{w^{n_1}}{z^{n_2}}\frac{1}{w-z},
\end{multline}
for any $\delta>0$. Also in this case, this kernel is a special case of the kernel (\ref{eq3.43}).

A way of proving is the following. Obviously, changing the condition $\lambda_k^{m+1} \leq \lambda_k^m$ into $\lambda_k^{m+1} < \lambda_k^m$ does not change the system, since we cut out null sets. Then, using Lemma~\ref{LemInterlacing} (that clearly holds also in the continuous) one replaces the interlacing condition by a product of determinants,
\begin{equation}
\prod_{m=2}^{N} \det[\phi(\lambda_i^{m-1},\lambda_j^{m})]_{i,j=1}^{m},
\end{equation}
where $\lambda^{m-1}_m\equiv{\rm virt}$ are \emph{virtual variables}, $\phi(x,y)=\Id_{[x\leq y]}$, $\phi({\rm virt},y)=1$. Thus, the measure on ${\cal D}^{(N)}$ becomes
\begin{equation}\label{eq22}
{\rm const} \times \left(\prod_{m=2}^{N} \det[\phi(\lambda_i^{m-1},\lambda_j^{m})]_{i,j=1}^{m}\right) \Delta_N(\lambda^N)\prod_{i=1}^N e^{-(\lambda_i^N)^2}\,\dx\lambda,
\end{equation}
where $\dx\lambda=\prod_{1\leq k\leq n\leq N}\dx\lambda_k^n$. At this point one applies Theorem~\ref{ThmMinors} and a few computations leads to the above result.

\subsection{Correlation functions along space-like paths}
The natural question, in view of what we made with the interlacing particle system, is whether one can combine the two above special cases and get the correlation functions for space-like paths as above.

There are two aspects to be considered. The first is to determine whether the evolution of the minors' eigenvalues can be described by a Markov process. It is known that it is not the case if one takes at least three consecutive minors~\cite{ANvM10b}. However, along space-like paths the evolution is indeed Markovian as shown in Section 4 of~\cite{FF10}.

The second issue concerns the correlation functions and if they have any similarities with the ones for the interlacing particle system defined above. The answer is affirmative.
For $n\in\{1,\dots,N\}$ we denote by $H(n,t)$ the $n\times n$ minor of $H(t)$, which is obtained by keeping the first $n$ rows and columns of $H(t)$. Denote by $\lambda^n_1(t)\leq \lambda_2^n(t)\leq \dots\leq \lambda_n^n(t)$ the eigenvalues of $H(n,t)$. Then, at any time $t$, the interlacing property (\ref{eq18}) holds. Moreover, along space-like paths the eigenvalues' process is Markovian with correlation functions given as follows.

\begin{thm}[Ferrari, Frings; Theorem 1.2 of~\cite{FF10}]\label{ThmMainGUE}
For any $m=1,2,\dots$, pick $m$ (distinct) triples
\begin{equation}
(n_j,t_j,x_j)\in \N\times \R_{\geq 0}\times\R
\end{equation}
such that $n_1\leq n_2\leq \ldots\leq n_m$ and $t_1\geq t_2\geq \ldots\geq t_m$. Then the eigenvalues' point process is determinantal. Its correlation kernel is given by
where
\begin{equation}\label{eqExtendedGUEkernel}
\begin{aligned}
K(n_1,t_1,x_1;n_2,t_2,x_2) =&-\frac{2}{2\pi\I}\int_{\I\R+\delta}\dx w\, \frac{e^{(t_1-t_2)w^2-2(x_1-x_2)w}}{w^{n_2-n_1}}\Id_{[(n_1,t_1)\prec (n_2,t_2)]}\\
&+ \frac{2}{(2\pi\I)^2}\oint_{|z|=\delta/2}\dx z \int_{\I\R+\delta}\dx w \, \frac{e^{w^2t_1-2x_1w}}{e^{z^2t_2-2x_2z}}
\frac{1}{w-z}\frac{w^{n_1}}{z^{n_2}}
\end{aligned}
\end{equation}
where $\delta>0$.
\end{thm}
Notice that this is exactly the kernel (\ref{eq3.43}) obtained above.

\subsection{Large time asymptotics}
The large $(n,t)$-asymptotics along space-like paths can be made by steep descent analysis as explained in Section~\ref{sectAsymptAnalysis}. Here we present a scaling with varying levels and times for which the Airy$_2$ process arises. For asymptotics at fixed time see also~\cite{FN08b}.

Consider the scaling
\begin{equation}
\begin{aligned}
n(u)&=\eta L-\alpha u L^{2/3},\\
t(u)&=\tau L+\beta u L^{2/3}.
\end{aligned}
\end{equation}
where $\eta>0,\tau>0$ and $\alpha,\beta\geq 0$. From the macroscopic picture one obtains that the smallest eigenvalue is at level $n(u)$ and time $t(u)$ is around $-\sqrt{2n(u) t(u)}$ plus fluctuations of order $L^{1/3}$ (for the largest eigenvalue it is similar). Thus consider the scaling
\begin{equation}
x(u,\xi)=-\sqrt{2n(u) t(u)} -\xi L^{1/3}.
\end{equation}
This scaling limit corresponds to consider the rescaled smallest eigenvalues
\begin{equation}
\lambda_t^{\rm resc}(u):=\frac{\lambda_{1}^{n(u)}(t(u))+\sqrt{2n(u) t(u)}}{-L^{1/3}}.
\end{equation}
The accordingly rescaled correlation kernel is therefore given by
\begin{equation}
K_L^{\rm resc}(u_1,\xi_1;u_2,\xi_2)=L^{-1/3} K(n(u_1),t(u_1),x(u_1,\xi_1);n(u_2),t(u_2),x(u_2,\xi_2))
\end{equation}
with $K$ given in Theorem~\ref{ThmMainGUE}. For the asymptotic analysis, the critical point is $z_c=\sqrt{\eta/2\tau}$. The change of variable $z=z_c+W (\kappa L)^{-1/3}$ and \mbox{$w=w_c+W (\kappa L)^{-1/3}$} with $\kappa=(2\tau)^{3/2}/\eta^{1/2}$ and the control of the error terms leads eventually to
\begin{equation}
\lim_{L\to\infty}K_L^{\rm resc}(u_1,\xi_1;u_2,\xi_2) \equiv S_v^{-1} K_2(s_1/S_v,u_1/S_h;s_2/S_v,u_2/S_h)
\end{equation}
with $S_h=(2\eta^{2/3}\tau)/(\alpha \tau +\beta \eta)$ and $S_v=\sqrt{\tau/2}/\eta^{1/6}$.

By further controlling the tails of the kernel one obtains the convergence of the Fredholm determinants. In particular, one obtains for the distribution of the smallest eigenvalues
\begin{equation}
\lim_{L\to\infty} \frac{\lambda_{1}^{n(u)}(t(u))+\sqrt{2n(u) t(u)}}{-L^{1/3}} = S_v {\cal A}_2(u/S_h)
\end{equation}
in the sense of finite-dimensional distributions.

\section{Final remarks}\label{SectOutlook}

\subsection*{Stochastic growth models}
The link between growth models and random matrices goes back to~\cite{Jo00b,PS00}. TASEP can be interpreted as a stochastic growth model, by defining a height function whose discrete gradient is given by $h(j+1)-h(j)=1-2\eta_j$ with $\eta_j$ the TASEP occupation variable. Then, TASEP is a model in the so-called KPZ universality class (where KPZ stands for Kardar-Parisi-Zhang~\cite{KPZ86}), see also the review~\cite{Cor11}.

As for random matrices, where the different symmetry for GOE and GUE ensembles have different limiting distributions, also for growth models the limiting law  differs depending on whether the initial condition is translation invariant or not, see the review~\cite{Fer10b} for more details. Another solvable model and well-studied model in the KPZ class is the polynuclear growth model, see~\cite{BFS07b,PS02,Jo03b}.

A different way of analyzing TASEP is to use the link to a directed last passage percolation (LPP) and then use the non-intersecting line ensembles approach~\cite{Jo03b,PS02,Jo02b}. For multi-point distributions, one can transfer the LPP results to distributions of particle positions or height function using the slow-decorrelation results~\cite{Fer08,CFP10b,CFP10a}. The non-intersecting line ensemble approach is successful only for ``step-like'' initial condition or two-sided Bernoulli, but not for ``flat'' initial condition. The latter can be analyzed using the interlacing structure too~\cite{BFPS06,BFS07b,BF07}.

\subsection*{Random tilings - Aztec diamond}
The Tracy-Widom distribution and the Airy$_2$ process shows up also in random tiling models like the Aztec diamond~\cite{Jo03}. It is not an accident. Indeed, the construction of the Markov chain explained in this lectures works also for a particle system that generates the uniform measure on Aztec diamond~\cite{BF08,Nor08} (this is used for the Java simulation in~\cite{FerAZTEC}, and is related to the shuffling algorithm~\cite{CEP96}).

\subsection*{The connection is only partial}
We saw in these lectures the situations where the eigenvalues' measure coincide with the measure on the interacting particle system. The connection is however restricted to space-like paths and for GUE. The GOE Tracy-Widom distribution~\cite{TW96} arises in TASEP for periodic initial condition for TASEP. However, the (determinantal) measure on interlacing points is not a positive measure, so that the fixed level measure is not a probability measure coming from some random matrix model. Also, it is known that the joint distributions of particle positions with periodic initial condition is asymptotically governed by the Airy$_1$ process~\cite{Sas05,BFPS06}, which does not give the joint distributions of the largest eigenvalue of GOE Dyson's Brownian motion~\cite{BFP08}.

\subsection*{Gaussian Free Field in the bulk of the particle system}
In these lectures we focused at the scaling limit for the edge of the interlacing particles. If we focus in the middle, one finds the sine kernel. Also, interpreting the system as a two-dimensional interface, it leads to a model in the $2+1$-dimensional anisotropic KPZ class~\cite{BF08,BF08b} (another model in this class is~\cite{PS97}). In the bulk the height function is Gaussian in a $\sqrt{\ln(t)}$ scale and one sees the Gaussian Free Field (GFF) as shown in~\cite{BF08} (see~\cite{She03,Ken01} to learn more about the GFF).

\subsection*{Random matrix process in the bulk of random matrices}
The fact that the particle process and the minor process in general are different, can be seen also from the works of Borodin where he obtain the analogue of the GFF for the minors of Wigner matrices~\cite{Bor10b,Bor10c}.

\subsection*{Intertwining in Macdonald processes}
Finally, the abstract formalism of~\cite{BF08} was carried over to a new ground of the so-called Macdonald processes and applied to difference operators arising from the (multivariate) Macdonald polynomials~\cite{BC11}. This lead to a conceptual new understanding of the asymptotic behavior of \mbox{(1+1)-dimensional} random polymers in random media and finding explicit solutions of the (nonlinear stochastic partial differential) KPZ equation~\cite{BC11,BCF12}.

\appendix

\section{Toeplitz-like transition probabilities}\label{AppToeplitz}
Here we give some results on transition probabilities which are translation invariant. They are taken from Section~2.3 of~\cite{BF08}. Our cases are obtained by limits when all $\alpha_i$'s goes to the same value, say $1$.

Recall that $W_n=\{x^n=(x_1,\ldots,x_n)\in\Z^n\, | \, x_1<x_2<\ldots<x_n\}$.
\begin{prop}[Proposition 2.8 of~\cite{BF08}]\label{propApp1}
Let $\alpha_1,\ldots,\alpha_n$ be non-zero complex numbers and let $F(z)$ be an analytic function in an annulus $A$ centered at the origin that contains all $\alpha_j^{-1}$'s. Assume that $F(\alpha_j^{-1})\neq 0$ for all $j$. Then, for $x^n\in W_n$,
\begin{equation}
\frac{\sum_{y^n\in W_n}\det\big(\alpha_i^{y_j^n}\big)_{i,j=1}^n \det(f(x_i^n-y_j^n))_{i,j=1}^n}{F(\alpha_1^{-1})\cdots F(\alpha_n^{-1})} = \det\big(\alpha_i^{x_j^n}\big)_{i,j=1}^n,
\end{equation}
where
\begin{equation}
f(m)=\frac{1}{2\pi\I}\oint_{\Gamma_0}\dx z \frac{F(z)}{z^{m+1}}.
\end{equation}
\end{prop}

A simple corollary for the specific case of the transition probability given in Proposition~\ref{propDiscreteCharlier} is the following.
\begin{cor}
For $F(z)=1-p+pz^{-1}$, it holds
\begin{equation}
f(m)=\frac{1}{2\pi\I}\oint_{\Gamma_0}\dx z \frac{F(z)}{z^{m+1}}
=\left\{
         \begin{array}{ll}
           p, & \textrm{if }m=-1, \\
           1-p, & \textrm{if }m=0, \\
           0, & \textrm{otherwise},
         \end{array}
       \right.
\end{equation}
and
\begin{equation}
\sum_{y^n\in W_n}\Delta_n(y^n) \det(f(x_i^n-y_j^n))_{i,j=1}^n=\Delta_n(x^n).
\end{equation}
\end{cor}
\begin{proof}
In the limit $\alpha_k\to 1$ for all $k=1,\ldots,n$ we have
\begin{equation}
\det\big(\alpha_i^{y_j^n}\big)_{i,j=1}^n /\det\big(\alpha_i^{x_j^n}\big)_{i,j=1}^n \to \Delta_n(y^n)/\Delta_n(x^n).
\end{equation}
\end{proof}

\begin{prop}[Proposition 2.9 of~\cite{BF08}]\label{propApp2}
Let $\alpha_1,\ldots,\alpha_n$ be non-zero complex numbers and let $F(z)$ be an analytic function in an annulus $A$ centered at the origin that contains all $\alpha_j^{-1}$ for $j=1,\ldots,n-1$ and assume that $F(\alpha_j^{-1})\neq 0$ for all $j=1,\ldots,n-1$. Let us set $y_n^{n-1}=\virt$ and $f(x-\virt)=\alpha_n^x$. Then
\begin{equation}
\frac{\sum_{y^{n-1}\in W_{n-1}}\det\big(\alpha_i^{y_j^{n-1}}\big)_{i,j=1}^{n-1} \det(f(x_i^n-y_j^{n-1}))_{i,j=1}^n}{F(\alpha_1^{-1})\cdots F(\alpha_{n-1}^{-1})} = \det\big(\alpha_i^{x_j^n}\big)_{i,j=1}^n.
\end{equation}
\end{prop}

A corollary concerning the transition kernel (\ref{eqMarkovLink}) is the following.
\begin{cor}\label{CorA4}
Let us choose $F(z)=(1-z)^{-1}$. Then
\begin{equation}
f(m)=\frac{1}{2\pi\I}\oint_{\Gamma_0}\dx z \frac{F(z)}{z^{m+1}}
=\left\{
         \begin{array}{ll}
           1, & \textrm{if }m\geq 0, \\
           0, & \textrm{otherwise},
         \end{array}
       \right.
\end{equation}
and, with $y_n^{n-1}=\virt$,
\begin{equation}
\sum_{y^{n-1}\in W_{n-1}}(n-1)!\Delta_{n-1}(y^{n-1}) \det(f(x_i^n-y_j^{n-1}))_{i,j=1}^n=\Delta_n(x^n).
\end{equation}
\end{cor}
\begin{proof}
We can get the result by considering first $1<\alpha_n<\alpha_{n-1}<\cdots<\alpha_1$ so that their inverse lie inside an annulus of radius less than $1$ (where the function $F$ is analytic) and then take the limit when all $\alpha_i$'s go to $1$. For instance, let $\alpha_k=(1-k\e)^{-1}$ for $k=1,\ldots,n$ (with $\e<1/n$). Then,
\begin{equation}
\prod_{k=1}^{n-1}\frac{1}{F(\alpha_k^{-1})}=(n-1)! \e^{n-1}.
\end{equation}
Further,
\begin{equation}
\det\big(\alpha_i^{y_j^{n-1}}\big)_{i,j=1}^{n-1}= \Delta_{n-1}(y^{n-1}) \e^{n(n-1)/2}(1+o(\e))
\end{equation}
so that as $\e\to 0$,
\begin{equation}
\frac{\det\big(\alpha_i^{y_j^{n-1}}\big)_{i,j=1}^{n-1}}{\det\big(\alpha_i^{x_j^n}\big)_{i,j=1}^n}= \frac{\Delta_{n-1}(y^{n-1})}{\Delta_n(x^n)\e^{n-1}} (1+o(\e))
\end{equation}
which leads to the result.
\end{proof}

By the above results, we can define the transition kernels
\begin{equation}
T_n(\alpha_1,\ldots,\alpha_n;F)(x^n,y^n) =\frac{\det\big(\alpha_i^{y_j^n}\big)_{i,j=1}^n}{\det\big(\alpha_i^{x_j^n}\big)_{i,j=1}^n}\frac{\det(f(x_i^n-y_j^n))_{i,j=1}^n}{F(\alpha_1^{-1})\cdots F(\alpha_n^{-1})}
\end{equation}
for $x^n,y^n\in W_n$, and
\begin{equation}
T^n_{n-1}(\alpha_1,\ldots,\alpha_n;F)(x^n,y^{n-1}) =\frac{\det\big(\alpha_i^{y_j^{n-1}}\big)_{i,j=1}^{n-1} }{\det\big(\alpha_i^{x_j^n}\big)_{i,j=1}^n}\frac{\det(f(x_i^n-y_j^{n-1}))_{i,j=1}^n}{F(\alpha_1^{-1})\cdots F(\alpha_{n-1}^{-1})}
\end{equation}
for $x^n\in W_n$ and $y^{n-1}\in W_{n-1}$. In our application, $T_n$ is $P_n$ and $T^n_{n-1}$ is $\Lambda^n_{n-1}$. The intertwining condition is then a consequence of the following result.
\begin{prop}[Proposition 2.10 of~\cite{BF08}]\label{propApp3}
Let $F_1$ and $F_2$ two functions holomorphic in an annulus centered at the origin and containing all the $\alpha_j^{-1}$'s that are nonzero at these points.
Then,
\begin{equation}
\begin{aligned}
&T_n(F_1) T_n(F_2)=T_n(F_2) T_n(F_1)=T_n(F_1 F_2),\\
&T_n(F_1) T^n_{n-1}(F_2)=T^n_{n-1}(F_1) T_n(F_2)=T^n_{n-1}(F_1 F_2).
\end{aligned}
\end{equation}
\end{prop}


\end{document}